%% file: arXiv1-per-node-sync.tex
\documentclass[a4paper,USenglish]{lipics-v2021}
\usepackage{amsmath}
\usepackage{amssymb}
\usepackage{algorithm}
\usepackage[noend]{algpseudocode}
\usepackage{afterpage}
\usepackage{amsthm}

\title{Randomized Tree-Intersection Leader Election}
\author{Yuval Emek}%
{Technion, Israel}%
{yemek@technion.ac.il}%
{https://orcid.org/0000-0002-3123-3451}
{Partially supported by an Israel Science Foundation (ISF) grant 730/24
and by The Bernard M. Gordon Center for Systems Engineering at the Technion.}
\author{Shay Kutten}%
{Technion, Israel}%
{kutten@technion.ac.il}%
{https://orcid.org/0000-0003-2062-6855}
{Partially supported by Israel Science Foundation (ISF) grant 1346/22
and by The Bernard M. Gordon Center for Systems Engineering at the Technion
and by the Technion's Grand Energy program.}
\author{Ido Rafael}%
{Technion, Israel}%
{ido.rafael@campus.technion.ac.il}%
{https://orcid.org/0000-0003-4923-4660}{}
\author{Gadi Taubenfeld}%
{Reichman University, Israel}%
{tgadi@runi.ac.il}%
{https://orcid.org/0000-0003-3070-5370}{}

\authorrunning{Y.~Emek, S.~Kutten, I.~Rafael, and G.~Taubenfeld}
\Copyright{Yuval Emek, Shay Kutten, Ido Rafael, Gadi Taubenfeld}

\nolinenumbers
\hideLIPIcs

\EventEditors{John Q. Open and Joan R. Access}
\EventNoEds{2}
\EventLongTitle{42nd Conference on Very Important Topics (CVIT 2016)}
\EventShortTitle{CVIT 2016}
\EventAcronym{CVIT}
\EventYear{2016}
\EventDate{December 24--27, 2016}
\EventLocation{Little Whinging, United Kingdom}
\EventLogo{}
\SeriesVolume{42}
\ArticleNo{33}

\def\false{{\it false}}
\def\true{{\it true}}
\def\flag{{\it flag}}

\begin{document}

\maketitle

\begin{abstract}
We present a randomized leader election algorithm for synchronous complete $n$-node graphs in the \textsf{CONGEST} model
that introduces a highly tunable trade-off between time complexity and the per-node message complexity.
By adjusting a single branching parameter, $\ell$, system designers can smoothly shift
the algorithmic burden from execution time to per-node message complexity, all while maintaining a strictly
sublinear total message complexity of $O(\sqrt{n} \log^{1.5} n)$.
We achieve this by utilizing dynamically truncated $\ell$-ary tree expansions coupled with a novel
``silent pulse'' verification mechanism. By forcing the expansions to form exact-volume almost-complete trees,
nodes can safely aggregate topological weights without overshooting the sublinear message bounds.
Specifically, our algorithm achieves $O(\log_\ell \sqrt{n \log n})$ time (round) complexity and $O(\ell)$ per-node message complexity.
This flexibility allows networks with tight bandwidth constraints to operate with a minimal $O(1)$ per-node burden, while high-bandwidth environments can collapse the election into $O(1)$ time units.
\end{abstract}

\section{Introduction}

Leader election is a fundamental symmetry-breaking problem in distributed computing,
serving as a vital primitive for tasks ranging from token passing to consensus and
spanning tree construction. In modern, large-scale networks, algorithms must not
only be efficient in terms of time and total communication, but they must also
reduce and balance the total communication burden placed on individual nodes.

%It is important to clarify the exact nature of the symmetry-breaking problem we address.
We focus on the \emph{implicit} leader election problem, where the fundamental requirement
is that exactly one node terminates in a leader state, while all others terminate in a non-leader state.
This is in contrast to \emph{explicit} leader election, which imposes the heavier requirement
that every node in the network must also learn the unique identifier of the elected leader.
While significant progress has been made in optimizing the
total message complexity and execution time of implicit leader election
\cite{GRS2018,KM2021,KumarMolla2023,KMPP2020,KuttenJACM2015,Kutten2015},
less effort has been made toward improving the total
communication burden placed on individual nodes.

While total message complexity captures the aggregate communication cost of an algorithm,
it often masks severe imbalances in how this burden is distributed across the network.
In many modern distributed systems -- such as wireless sensor networks, IoT deployments,
and ad-hoc mobile networks -- individual nodes operate under strict energy and local bandwidth constraints.
An algorithm with an optimal total message complexity might still overwhelm a small subset of nodes,
rapidly depleting their batteries or creating severe local congestion hotspots that degrade overall network throughput.
Consequently, bounding the per-node message complexity is an important
requirement for ensuring equitable resource consumption, avoiding single points of failure,
and extending the operational lifetime of the network.

Beyond constraining the processing and energy burden on individual nodes,
bounding the total per-node message complexity inherently restricts overall link
congestion throughout the algorithm's execution. While the standard \textsf{CONGEST}
model limits the bandwidth on any edge \textit{per round}, an algorithm with a total
per-node message complexity of $O(X)$ strictly guarantees a total per-edge message
complexity of $O(X)$ across all rounds. Since the converse does not hold---a network
may maintain low per-edge traffic while still overwhelming high-degree nodes---optimizing
for per-node complexity serves as a strictly stronger condition that prevents both
node-level and link-level bottlenecks over the entire execution.

In this paper,
we present a randomized leader election algorithm for synchronous complete $n$-node
graphs that introduces a highly tunable trade-off between time complexity and
per-node message complexity.
We achieve this through a novel ``tree-intersection'' strategy.
Rather than relying on fixed-topology message passing, our algorithm allows
candidates to initiate an $\ell$-ary tree expansion, and
ensure that the trees initiated by different candidates
intersect with high probability.

However, a major technical hurdle arises when $\ell$ is large: a naive, perfect tree
expansion could overshoot the required footprint by a factor of $\ell$, ruining the
algorithm's sublinear communication guarantees.%
\footnote{The stopping condition for a perfect tree is very simple because
all leaves reside on the last level.}
We overcome this by forcing the
expansion to build an \emph{almost-complete tree}.%
\footnote{An almost complete tree is a tree structure where every level is completely filled,
except possibly the last level. Thus, all leaves reside on the last two levels.}
By embedding an explicit leaf-allocation
budget within messages, nodes dynamically truncate the expansion the
precise round the exact target footprint is satisfied.

To achieve time and message efficiently (and simplicity)
within the strict $O(\log n)$-bit bandwidth limit of the \textsf{CONGEST} model,
we introduce two novel algorithmic mechanisms: the ``silent pulse'' optimization and a ``conservation of weight'' principle. Instead of explicitly tracking duplicate branches or pruned paths,
our framework relies on strict layer-by-layer synchronization.
A parent node expects an acknowledgment from its children at a precisely calculated round.
If a branch is pruned by a higher-rank candidate or absorbed due to an internal collision,
it silently drops out of the communication flow. By treating non-arriving messages as zero weight,
the algorithm securely aggregates acknowledgments without requiring explicit rejection messages,
maintaining the structural integrity of the tree.

We note a fundamental topological barrier: it is impossible to design an algorithm that achieves
this highly tunable trade-off between time and per-node message complexity for arbitrary network topologies.
This limitation is best illustrated by a double star graph.%
\footnote{
A graph composed of two star graphs whose central vertices are
connected by an edge is called a double star graph.
We notice that in a star graph the single center node can simply elect itself.}
In a double star topology, any symmetry-breaking process
that requires a total message complexity of $M$ inherently forces the two center nodes
to act as a routing bottleneck. Consequently, the two center nodes must process a fraction
of the messages proportional to $\Omega(M)$.
Because the local load is dictated by the graph's structure rather than the algorithm's design,
the per-node message complexity cannot be decoupled from the total message complexity,
preventing the existence of a universal, topology-independent algorithm with
a tunable trade-off between time and per-node message complexity.

The branching factor, $1\leq \ell\leq 2\sqrt{n \log n}$, serves as the algorithmic dial.
As summarized in Table~\ref{tab:complexity_frontier}, scaling $\ell$
allows to smoothly shift the computational burden
from time to per-node complexity, all while maintaining a
sublinear total message complexity of $O(\sqrt{n} \log^{1.5} n)$.

\begin{table}[htbp]
\centering
\caption{The Complexity Frontier: Trade-offs based on Branching Factor $1\leq \ell\leq 2\sqrt{n \log n}$}
\label{tab:complexity_frontier}
\renewcommand{\arraystretch}{1.4}
\begin{tabular}{|l|c|c|c|}
\hline
\textbf{Branching Factor ($\ell$)} & \textbf{Time Complexity} & \textbf{Per-Node Messages} & \textbf{Total Messages} \\ \hline
$\ell = 1$ (Walk) & $O(\sqrt{n \log n})$ & $O(1)$ & $O(\sqrt{n} \log^{1.5} n)$ \\ \hline
$\ell = 2$ (Binary Tree) & $O(\log \sqrt{n \log n})$ & $O(1)$ & $O(\sqrt{n} \log^{1.5} n)$ \\ \hline
$\ell = 2\sqrt{n \log n}$ (Star) & $O(1)$ & $O(\sqrt{n \log n})$ & $O(\sqrt{n} \log^{1.5} n)$ \\ \hline
\textbf{General $\ell$ $(\ell\geq 2$)} & $\mathbf{O(\log_\ell \sqrt{n \log n})}$ & $\mathbf{O(\ell)}$ & $\mathbf{O(\sqrt{n} \log^{1.5} n)}$ \\ \hline
\end{tabular}
\end{table}

While branching factors of $\ell \ge 2$ offer exponential speedups in latency, the baseline case of $\ell=1$ remains crucial for strictly constrained networks. Although both $\ell=1$ and $\ell=2$ share the same $O(1)$ per-node message complexity, the $\ell=1$ configuration yields smaller constant factors. By relying on a single propagating path rather than coordinating multiple concurrent branches, it strictly minimizes local memory state and bandwidth overhead. This minimalist extreme is the starting point of the tunable spectrum our framework provides.

The parameterized approach we use frames our work not merely as a new algorithm,
but as a generalization of the state-of-the-art for complete graphs.
Prior innovative work by Kutten et al. \cite{Kutten2015} established strict
sublinear bounds for randomized leader election.
Their algorithm achieves an impressive $O(1)$ time complexity,
but inherently requires an $O(\sqrt{n \log n})$
per-node message complexity.
Viewed through the lens of our framework, the algorithm in \cite{Kutten2015} for complete graphs
operates at the extreme edge of our complexity frontier, analogous to setting
the branching factor to $\ell = 2\sqrt{n \log n}$.

Conversely, our baseline path-intersection algorithm ($\ell = 1$, first line of Table~\ref{tab:complexity_frontier}) represents the opposite extreme. By introducing the $\ell$-ary tree expansion,
we unlock the entire continuous spectrum between these two extremes.
This flexibility allows networks with tight bandwidth or energy constraints
(such as ad-hoc or sensor networks) to operate with a minimal $O(1)$ per-node burden,
while high-bandwidth environments with less energy constraints (such as data centers)
can collapse the election into constant time.

The remainder of this paper is organized as follows.
Section~\ref{sec:ModelPreliminaries} formally defines the computational model.
In Section~\ref{sec:Path-Intersection:Complete Graphs}
we present the \emph{path-intersection} algorithm
%(Algorithm \ref{alg:path-intersection1})
which is a simplified special
version of the \emph{tree-intersection} algorithm.
%(presented in Section \ref{sec:Tree-Intersection:CompleteGraphs}),
%and is used to demonstrate some of the key ideas of our approach.
In Section \ref{sec:Tree-Intersection:CompleteGraphs},
we introduce the tree-intersection algorithm.
Section~\ref{sec:related_work} situates our contribution within the broader literature,
and finally, Section~\ref{sec:discussion} concludes the paper and discusses open questions for future work.

\section{Model and Preliminaries}
\label{sec:ModelPreliminaries}

%Before detailing our parameterized tree-intersection mechanism, we formalize the network environment,
%the communication constraints, and the complexity measures that govern the algorithm's execution.\\
%\\
\textbf{Communication Model.}
We model the communication network as an undirected, connected complete graph $G=(V,E)$, where $V$
is the set of nodes with $|V|=n$ and $E$ is the set of communication links.
We assume the system is fully synchronous, meaning execution proceeds in discrete,
globally synchronized rounds.

To avoid ambiguity, we explicitly define the anatomy of a synchronous round in our model.
Each round consists of two consecutive steps: first, a \emph{send} step where each node
evaluates its local state and transmits messages to its neighbors; second, a
\emph{receive} step where the messages sent in the current round arrive and are
processed by the node to update its state.

Communication strictly adheres to the \textsf{CONGEST} model: in the send step of each round,
any node can send a potentially distinct message to each of its neighbors in $G$,
provided the size of each message is strictly bounded by $O(\log n)$ bits.
Each node $v \in V$ is anonymous. During the algorithm execution,
each node draws, from a polynomial space, a unique identifier $ID_v$ (unique w.h.p.).
All nodes have knowledge of $n$.\\
\\
\textbf{Complexity Measures.}
We evaluate the efficiency of our algorithms using the following three standard metrics:
\begin{itemize}
\item
\emph{Time Complexity:}
The total number of synchronous communication rounds required from the initiation
of the algorithm until the final node safely terminates and the system reaches its final state.
\item
\emph{Total Message Complexity:} The aggregate number of $O(\log n)$-bit messages
transmitted across all communication links in the network
over the entire execution of the algorithm.
\item
\emph{Per-Node Message Complexity:} The maximum number of $O(\log n)$-bit
messages processed (sent or received) by any single individual node throughout
the algorithm's execution.
%This metric strictly bounds the worst-case
%local bandwidth and energy overhead imposed on any individual network participant.
\end{itemize}
$~$\\
\textbf{Implicit Leader Election.}
Initially, all nodes are in a uniform initial state.
The goal of the algorithm is to break symmetry and elect a single leader.
An algorithm successfully solves implicit leader election if, upon termination of the algorithm,
the following conditions hold: (1) Exactly one node terminates in the leader state,
and (2) All other nodes terminate in a non-leader state.
Crucially, there is no requirement for nodes terminating in the non-leader state
to know the $ID$ of the elected leader \cite{Lyn96}.
This termination condition avoids the trivial $\Omega(n)$ message complexity lower bound
known for explicit leader election.\\
\\
\textbf{The notion ``with high probability''.}
For simplicity, since we generalize results from  \cite{Kutten2015},
we will use the same definition used in \cite{Kutten2015}.
An event occurs with high probability if, for some constant $c \ge 1$,
the probability of the event occurring is at least: $1 - n^{-c}$.
Our results can easily be adapted
(by scaling up the constant coefficients in our sample sizes and branch volumes accordingly)
to satisfy the following stronger definition \cite{EKRT2025,MU2017}:
An event $A$ occurs with high probability if
$\Pr (A) \geq 1 - n^{-c}$ for any (arbitrarily large) constant $c \ge 1$.

%\paragraph*{The Notion ``With High Probability''.}
%Throughout this paper, we say that an event occurs \emph{with high probability} (w.h.p.)\ if,
%for any arbitrarily large constant $c \ge 1$, the probability of the event occurring is at
%least $1 - n^{-c}$. This standard convention (see, e.g., \cite{EKRT2025, MitzenmacherUpfal})
%guarantees that the error probability of our randomized sub-routines can be made polynomially
%small by adjusting the constant factors in our parameters. We note that our results can
%be seamlessly adapted to softer definitions found in the literature—such as requiring
%a fixed success probability of $1 - 1/n^c$ for a specific choice of $c \ge 1$
%(as seen in \cite{Kutten2015})—by scaling down the constant coefficients in our sample sizes and branch volumes %accordingly.

%\newpage
\section{Randomized Path-Intersection Leader Election}
\label{sec:Path-Intersection:Complete Graphs}
The \emph{path-intersection} algorithm,
%(Algorithm \ref{alg:path-intersection1}),
presented in this section, is a simplified, special version of the more complex \emph{tree-intersection} algorithm presented in the next section.
We present the path-intersection approach separately for two key reasons.
First, it serves as a pedagogical stepping stone, demonstrating the core random-walk intersection and trace-back mechanics without the added complexities of internal collisions and weight aggregation.
Second, in strictly constrained networks where $\ell=1$ is the optimal operating point, Algorithm~\ref{alg:path-intersection1} provides better constant factors in terms of local memory and processing overhead.
%as nodes only track a single propagating path rather than coordinating multiple concurrent branches.

In Algorithm \ref{alg:path-intersection1},
%operates in a synchronous undirected complete graph with $n$ nodes.
initially, each node independently becomes a \textit{candidate} with probability $p = \frac{\log n}{n}$.
This probability ensures that at least one and at most $4\log n$ nodes become candidates w.h.p.
Then, each candidate $v$ chooses a random rank $r_v \in \{1, \dots, n^4\}$, which ensures that the
ranks are unique w.h.p.\ due to the domain size $n^4$.

Candidates are whittled down through intersecting random walks.
Each candidate creates a token with a unique rank which starts a random walk.
To stay within the \textsf{CONGEST} bandwidth limits,
every node only agrees to pass along the single highest-ranked token it has seen so far,
effectively disqualifying smaller competitors on the spot.
Only the token that manages to be the \emph{maximum} at every single step
of its random walk will reach the end of the walk and find a path (of stored edges)
to follow back to the candidate that generated it.
If a candidate sees its own token returns,
%after this token completed a full length random walk,
it knows it is the global winner.

The full length of the random walks is $L = 2\times\lceil\sqrt{n \log n}\rceil -1$,
and it ensures that the walks are long enough to encounter one another w.h.p.
To ensure the global winner is correctly identified even if ranks intersect at different times,
a returning token is only permitted to continue
its journey if its rank remains the highest ever seen by every node on its path,
discarding any earlier tokens that were surpassed by a late-arriving superior rank.

The advantage of the path-intersection algorithm is its very low per-node message complexity,
only $O(1)$ messages, which improves upon the previously best known per-node message complexity
of $2\times\lceil\sqrt{n \log n}\rceil$ \cite{Kutten2015}.
Its core limitation is its sequential nature which leads to high time complexity,
$2\times\lceil\sqrt{n \log n}\rceil$ rounds, compared to one round of the algorithm of \cite{Kutten2015}.
Both algorithms (\cite{Kutten2015} and ours) have the same $O(\sqrt{n} \log^{1.5} n)$ total message complexity.

\subsection{The algorithm}
%Algorithm~\ref{alg:path-intersection1} appears in the sequel.
%Algorithm~\ref{alg:path-intersection1} appears below.
%\newpage
\begin{algorithm}[ht!]
\caption{Path-Intersection Leader Election in Complete Graphs}
\label{alg:path-intersection1}
%\small
\footnotesize
\begin{algorithmic}[1]
\Statex \textbf{The following code is executed by each node $v \in V$}
\Statex \textbf{Constants (known to all nodes):}
\Statex $G=(V, E)$ where $|V|=n$ \Comment $G$ is a complete graph
\Statex $L = 2\times\lceil\sqrt{n \log n}\rceil -1$ \Comment length of a random walk
\Statex \textbf{State Variables at node $v$:}
\State $r_{max} \gets 0$; $r_{v} \gets 0$; $\flag \gets \false$; $status \gets \bot$;
$prev\_hop[1 \dots L] \gets \bot$
\Comment{initialization}
\Statex \textbf{Phase 1: Candidate Selection (Round 0)}
\State Node $v$ independently becomes a \textit{candidate} with probability $p = \frac{\log n}{n}$
\If{$v$ is a candidate}
    \State $r_v\gets$ generate a random rank from $\{1, \dots, n^4\}$ \Comment{unique rank w.h.p.}
    \State $T_v\gets$ create token $\langle r_v \rangle$; $r_{max} \gets r_{v}$; $\flag \gets \true$
           \Comment{ready to forward $T_v$}
    \State $T_{best} \gets T_v$ \Comment{$T_{best}$ and $T_v$ refer to the \emph{same} token}
\EndIf

\Statex \textbf{Phase 2: Forward Random Walk (Rounds $1$ to $L$)}
\For{round $i = 1$ \textbf{to} $L$} \Comment{global rounds $1$ to $L$}
    \If{$\flag = \true$} \Comment{one step of the random walk}
         \State Select edge to neighbor $u \in \mathcal{N}(v)$ u.a.r. and forward $T_{best}$ to $u$; $\flag \gets \false$
    \EndIf

    \If{$v$ receives one or more tokens in the current round} \Comment{updates $T_{best}$}
        \State $T_{best}\gets $ the token with maximum rank among the received tokens
        % \Comment{the other tokens are dropped}
    \If{$r_{best} \geq r_{max}$} \Comment{$r_{best}$ denotes the rank of $T_{best}$}
        \State $r_{max} \gets r_{best}$; $\flag \gets \true$ \Comment{ready to forward $T_{best}$}
        \State $prev\_hop[i] \gets \text{arrival edge from sender of } T_{best}$ \Comment{facilitate trace-back}
    \Else
        ~\textbf{drop} $T_{best}$ \Comment{disqualified by a higher rank that arrived earlier}
    \EndIf
    \EndIf
\EndFor

\Statex \textbf{Phase 3: Trace-back and Verification (Rounds $L+1$ to $2L$)}
\For{round $i = 1$ \textbf{to} $L$} \Comment{global rounds $L+1$ to $2L$}
    \If{$\flag = \true$}
         \State Forward $T_{best}$ via edge $prev\_hop[L-i+1]$; $\flag \gets \false$\Comment{a retrace step}
    \EndIf

    \If{$v$ receives token $T_u$ in the current round} \Comment{$v$ can receive at most one token!}
    \If{$r_{u} = r_{max}$} \Comment{it is not possible that $r_{u} > r_{max}$}
        \State $T_{best}\gets T_u$; $\flag \gets \true$ \Comment{ready to forward $T_{best}$}
    \Else
        ~\textbf{drop} $T_{best}$ \Comment{disqualified by a higher rank that arrived later}
    \EndIf
    \EndIf
\EndFor

\Statex \textbf{Phase 4: Termination  (Round $2L+1$)}
\If{$v$ is a candidate \textbf{and} $v$ receives its own token $T_v$ in Round $2L$ \textbf{and} $r_{v} = r_{max}$}
     \State $status \gets$ \textbf{LEADER}
     \textbf{else} $status \gets$ NOT\_LEADER
\EndIf
\end{algorithmic}
\end{algorithm}
%
%\begin{remark}[Handling cycles during random walks]
%By utilizing the index $i$ to remember the previous hop,
%we have ensured that the algorithm is robust against the ``tail-biting'' nature of random walks.
%A node needs to know which entry in its history corresponds to the specific step of the journey being reversed.
%As formalized in Lemma~\ref{lem:unambiguous_retraction},
%this maintains the elegance of the synchronous model
%by treating the walk as a path in a time-expanded graph.
%\end{remark}

%\newpage
\subsection{Correctness and Complexity Analysis of Algorithm~\ref{alg:path-intersection1}}
\label{sec:analysis1}
%
%Next, we analyze the correctness and complexity of the \emph{Randomized Path-Intersection}
%algorithm.
%\subsubsection{Correctness}
%
At the heart of our approach is the well-known combinatorial phenomenon often
referred to as the Birthday Paradox \cite{MU2017}.
This phenomenon inspired the following lemma
which captures the high probability of intersection between any two random walks of sufficient length.

\begin{lemma}[Random Walk Intersection Probability]\label{lem:RW-intersection}
In a complete graph $K_n$, let $W_1$ and $W_2$
be two independent random walks generated by Algorithm~\ref{alg:path-intersection1}, each of length
$L = 2\left\lceil \sqrt{n\log n}\right\rceil - 1$.
The probability that $W_1\cap W_2=\emptyset$ is at most $n^{-4}$.
\end{lemma}

\input{Ido-0-alg-1-intersection.tex}

\begin{lemma}[Unambiguous Path Retraction]
\label{lem:unambiguous_retraction}
For any token $T$ that completes Phase 2 at round $L$,
the sequence of edges traversed in Phase 3 is the exact reverse of the sequence
of edges traversed in Phase 2, regardless of possible cycles.
\end{lemma}

\begin{proof}
Let the forward path of $T$ be represented as a sequence of nodes $P = (v_0, v_1, \dots, v_L)$.
In each round $i \in \{1, \dots, L\}$, node $v_i$ receives the token from $v_{i-1}$ and stores
the arrival edge in $prev\_hop[i]$.
If the path contains a loop such that $v_a = v_b$ for $a \neq b$, the node $v_a$ will
store two distinct arrival edges in two distinct memory locations: $prev\_hop[a]$ and $prev\_hop[b]$.
In Phase 3, during round $j \in \{1, \dots, L\}$, the node currently holding the token uses the index
$k = L-j+1$ to select its next hop. Since $j$ increments linearly, $k$ decrements linearly from $L$ to $1$.
This ensures that at any round $j$, the node $v_{L-j+1}$ retrieves exactly the edge $prev\_hop[L-j+1]$
that points back to $v_{L-j}$, effectively reconstructing the path $P' = (v_L, v_{L-1}, \dots, v_0)$.
Because the indexing is tied to the global synchronous round,
the physical identity of the nodes does not cause ambiguity.
\end{proof}

\begin{theorem}[Safety]
\label{thm:safety1}
With high probability, the algorithm elects at most one leader.
\end{theorem}

\begin{proof}
Let $C$ be the set of candidates. If $|C|=0$, no leader is elected.
If $|C| \geq 1$, let $v_{max}$ be the candidate with the unique maximum rank $rank_{max}$ (ranks are unique w.h.p. due to the domain size $n^4$).%
\footnote{The probability that any two of the $n$ nodes select the same rank is bounded by
$\binom{n}{2} \frac{1}{n^4} < \frac{1}{2n^2}$.
Thus, even if every node selects a rank, all ranks are unique with high probability.}
Suppose for contradiction a candidate $u \neq v_{max}$ with rank $r_u < rank_{max}$ is elected. For $u$ to be elected, its token $T_u$ must successfully complete the Trace-back Phase. This implies that for every node on $u$'s path, the condition $r_u = r_{max}$ held true during the retrace steps.
However, by Lemma~\ref{lem:RW-intersection}, the random walk of $u$ and the random walk of $v_{max}$ intersect at some node $z$ with high probability.
Furthermore, the phases ensure that all ``forward'' messages happen before all ``rebound'' messages.
When $v_{max}$'s token $T_{max}$ visited $z$ (during Phase 2), node $z$ updated $r_{max} \leftarrow rank_{max}$.
By Lemma~\ref{lem:unambiguous_retraction}, $T_u$ must retrace through $z$.
Since $rank_{max} > r_u$, any subsequent attempt by $T_u$ to retrace through $z$ would fail the check
$r_u = r_{max}$ at node $z$, causing $T_u$ to be dropped. Thus, $u$ cannot be elected.
\end{proof}

\begin{theorem}[Liveness]
\label{thm:liveness1}
With high probability, exactly one leader is elected.
\end{theorem}

\begin{proof}
No candidate is selected with probability $\approx (1-p)^n$, negligible for $p = \frac{\log n}{n}$.
The candidate $v_{max}$ with the global maximum rank $rank_{max}$ will never have its token dropped.
In Phase 2, for any node $w$ visited by $T_{max}$, the local variable $r_{max}$ of $w$
is updated to $rank_{max}$ because no rank
is strictly greater than $rank_{max}$. In Phase 3, the condition $rank_{max} = r_{max}$ will hold at every
step of the trace-back. Thus, $T_{max}$ returns to $v_{max}$, and $v_{max}$ is elected.
Combined with Theorem~\ref{thm:safety1}, exactly one leader is elected.
\end{proof}

\begin{theorem}[\textsf{CONGEST}]
\label{thm:CONGEST1}
The constraint of the \textsf{CONGEST} model is satisfied.
\end{theorem}

\begin{proof}
If a token $T_u$ is not the maximum at some point in its forward random walk, it is dropped.
Thus, in such a case, $T_u$ never starts the Trace-back Phase and never returns to the candidate.
By filtering for $T_{best}$ in Phase 2, we reduce the message complexity
to at most one message per node per round.
In Phase 3, a token $T$ will return to a node $v$ at a given round, say $i$,
only if it was submitted by
$v$ $2(L-i)$ rounds earlier, which implies that a node will not receive more than
one token in any round during Phase 3.
Thus, a node can only send one $O(\log n)$-bit message per round per edge.
\end{proof}

%\subsubsection{Complexity Analysis}
%
\begin{theorem}[Complexity Bounds]\label{thm:RW:Complexity Bounds}
Algorithm~\ref{alg:path-intersection1} achieves the following bounds w.h.p.:
\begin{enumerate}
    \item \textbf{Time Complexity:}
    $4\times \left\lceil\sqrt{n \log n}\right\rceil - 2 = O(\sqrt{n \log n})$
    rounds.
    \item \textbf{Total Message Complexity:}
    $O(\sqrt{n} \log^{1.5} n)$
    messages.
    \item \textbf{Per-Node Message Complexity:}
    $O(1)$
    messages.
\end{enumerate}
\end{theorem}
\input{Ido-1-alg-1-complexity.tex}

%\newpage
\section{The Generalized Tree-Intersection Framework}
\label{sec:Tree-Intersection:CompleteGraphs}
To achieve a dramatic speedup in time complexity, at the expense of per-node message complexity,
we generalize the ``path-intersection'' concept into a ``tree-intersection'' strategy.
Instead of sending a single token on a random walk, each candidate initiates
an $\ell$-ary tree expansion of size $2\times\lceil\sqrt{n \log n}\rceil$, which
ensures that trees encounter one another w.h.p.

\subsection{An Informal Description of the Tree-Intersection Strategy}
\begin{enumerate}
\item
\emph{Trading bandwidth for time.}
The advantage of the path-intersection approach is the low per-node message complexity,
its core limitation is its sequential nature.
By instructing each node to forward the rank to $\ell$ neighbors instead of just one,
we transform the linear walk into a tree.
This allows us to cover the same ``footprint'' required for intersection
in $O(\log_\ell \sqrt{n \log n})$ rounds (for $\ell \ge 2$).
For a binary tree ($\ell=2$), this is logarithmic time; for larger branching factors,
it collapses toward constant time.

\item
\emph{The all-back verification principle.}
While the ``Birthday Paradox'' ensures intersection between every two trees,
it does not guarantee that \emph{every} path will intersect with the global maximum.
We address this issue as follows.
A candidate is elected only if it
receives acknowledgments from all the $N = 2\times\lceil\sqrt{n \log n}\rceil$ (logical) \emph{nodes} generated.
If a branch encounters a node
already dominated by a higher rank, that branch is ``pruned''.
Crucially, at least the node at which the intersection happens
sends no acknowledgment to the lower rank candidate.
This prevents a lower rank candidate from accumulating the
necessary number of acknowledgments to become a leader.
%This is true regardless of whether the higher rank token
%arrives before after or at the same time as the lower ranked token.
\item
\emph{Handling internal collisions.}
A subtle challenge arises because we select neighbors randomly (with replacement):
a tree might loop back on itself, or two branches might converge on the same node.
%To address this, we assign a \emph{target weight} that a tree should achieve.%
%The target weight is exactly the number of node of an \emph{almost complete tree}
%of size $N = 2\times\lceil\sqrt{n \log n}\rceil$ constructed during the tree expansion.
When a node is reached by multiple (say $x>1$) parents of the same tree in a given round (an internal collision),
it treats only the first one as its true parent for this round,
%(For each round a node may have at most one parent.)
sets its own weight to $x$,
and sends \emph{explore} messages to $\ell \times x$  of its neighbors.
Later when ACKs arrive, the node reports the aggregated sum of the ACKs
plus its own weight $x$ to the single true parent.
This ensures that the total weight reported to the
highest ranked candidate matches the target weight $N$.
\item
\emph{The silent pulse optimization.}
Because the tree grows layer-by-layer, every parent knows exactly when to expect an ACK from its children.
If a parent hears nothing from a child in the expected round,
it assumes that child was a \emph{duplicate} branch (i.e., an internal collision) handled by someone else,
or a pruned branch.
This allows us to eliminate explicit ``I am a duplicate'' messages,
relying instead on the ``conservation of weight'' carried by the active branches.
%This again ensures that the total weight reported to the highest ranked candidate
%matches the target weight.
\item
\emph{Satisfying the \textsf{CONGEST} constraint.}
Every node in the graph stores and propagates only the messages (both Explore and ACK) of
the highest rank it has seen,
effectively pruning branches of all but the most dominant tree as the expansions collide.
The Explore and ACK messages have three counters associated with them and as we prove,
each one of these counters require at most $O(\log n)$ bits to encode.
\item
\emph{Optimizing tree volume via almost complete tree expansion.}
A naive implementation of the $\ell$-ary expansion might force the algorithm to construct a \emph{perfect} $\ell$-ary tree (because of the simple stopping condition).
However, this introduces an undesirable multiplicative overhead:
the final level could contain up to a factor of $\ell$ more nodes compared to the desired
volume of $2\times\lceil\sqrt{n \log n}\rceil$, blowing up the total message complexity by a factor of $\ell$.
To prevent this volume explosion, our framework constructs an \emph{almost complete $\ell$-ary tree}.
The expansion halts dynamically so that the total number of (logical) nodes is exactly $2\times\lceil\sqrt{n \log n}\rceil$.
Implementing the exact \textit{stopping condition} so that the expansion ``knows'' when to halt
is one of the key ideas of our design, and is explained in
Subsection~\ref{sec:logical_tree_construction}.
\end{enumerate}

%\newpage
\subsection{How the almost-complete tree is constructed}
\label{sec:logical_tree_construction}
\newcommand{\LeafNumber}{\operatorname{Leaf}}
To avoid the multiplicative volume overhead of a perfect $\ell$-ary tree,
our algorithm dynamically truncates the tree's expansion.
We achieve this by allowing physical network nodes to concurrently simulate
multiple independent \emph{logical} tree nodes, only compressing their data
into bundled physical transmissions when random draws route them over the same communication link.
By \emph{logical tree} we mean the conceptual tree constructed by a candidate.
If two branches converge on the same physical node, this physical node hosts multiple distinct logical nodes.
The construction builds a \emph{logical} almost complete $\ell$-ary tree
with $N$ logical nodes.

Let $N = 2\left\lceil\sqrt{n\log n}\right\rceil$
be the exact target volume of the logical tree, and let $1\le \ell\le N$
be the branching parameter.
For any integer $h\ge 0$, let $T_h$ be the volume of a perfect $\ell$-ary tree of height $h$.
That is, if $\ell=1$, then $T_h=h+1$, and if $\ell\ge 2$, then
$T_h=1+\ell+\ell^2+\cdots+\ell^h=\frac{\ell^{h+1}-1}{\ell-1}$.\\
\\
\textbf{The Core Tree and the Leaf Budget.}
We define $H$ as the maximum integer $h\ge 0$ such that $T_h<N$.
If $\ell=1$, then $H=N-2$, and if $\ell\ge 2$,
then $H=\left\lceil\log_\ell((\ell-1)N+1)\right\rceil-2$.
The perfect $\ell$-ary tree of height $H$ is called the \emph{core tree}.
To reach exactly $N$ total nodes, we must add
an exact number of leaves at the last level below the core tree.
We define this remaining requirement as the node's initial budget:
$\LeafNumber(\ell,N):=N-T_H$.\\
\\
\textbf{Token Bundling and Budget Distribution.}
Random walks may cause multiple branches of the same candidate's logical tree to collide at a single physical node
$v$. If $v$ hosts $t$ logical nodes belonging to the same candidate at core depth $d$,
it also holds an aggregated budget $x$ representing the total number of leaves at the last level those
$t$ tokens are allowed to spawn.
The construction begins at the root, which hosts exactly one depth-$0$ logical node $(t=1)$,
with initial budget $x= \LeafNumber(\ell,N)=N-T_H$.

To continue the expansion while respecting \textsf{CONGEST} constraints, $v$
must simulate the next step for all $t$ tokens.
If the expansion is still within the core tree $(0 \le d <H)$,
$v$ generates $\ell \cdot t$ logical children.
Each hosted depth-$d$ logical node becomes the logical parent of exactly $\ell$
of these children, and each such child has depth $d+1$.
Each of the newly generated $\ell\cdot t$ children corresponds to an independent random draw
(with replacement) from the physical neighbors of $v$.
Node $v$ partitions the total budget $x$ sequentially among these
$\ell\cdot t$ draws. Each draw is assigned the maximum capacity of $\ell^{H-d}$
as long as sufficient budget remains.
If the remaining budget is positive but strictly less than $\ell^{H-d}$,
the next draw receives this exact remainder, and all subsequent draws are assigned a budget of $0$.

Crucially, physical bundling occurs at the outbound edge level based on these random draws.
If a specific physical neighbor $u$ is selected $k$ times during the
draws, $v$ bundles these $k$ logical children into a \emph{single} physical message sent to $u$.
This bundled message carries the candidate's rank, the token count $k$,
and the sum of the specific leaf budgets assigned to those $k$ draws.

If the expansion has reached the boundary of the core tree $(d=H)$ at a physical node
$v$ that hosts $t$ depth-$H$ logical nodes and has budget $x$,
then the $t$ logical nodes do not spawn full branches.
Instead, $v$ partitions the total budget
$x$ sequentially among these $t$ nodes.
Each hosted logical node becomes the logical parent of $\ell$
leaves at the last level as long as sufficient budget remains.
If the remaining budget is positive but strictly less than $\ell$,
the next hosted logical node receives this exact remainder, and all
subsequent hosted logical nodes receive $0$ budget.
Overall, $v$ creates exactly $x$ leaves at the last level (each with budget 0)
for its hosted logical nodes.
Each leaf at the last level is then placed by an independent random draw
from the physical neighbors of $v$,
again bundling tokens and their respective budgets
into single physical messages if the same physical neighbor is drawn multiple times.\\
\\
\textbf{Message Complexity Overhead.}
Because the network size $n$ and parameter
$\ell$ are globally known, the structural constants $N$, $H$,
and the initial budget can be computed locally by all nodes.
The global synchronous clock provides the current depth $d$.
Therefore, the only dynamic state a bundled physical message must transmit
across any specific edge is the candidate's rank, the subset of logical
tokens routed over that edge (e.g., $k$),
and the sum of the specific leaf budgets assigned to those $k$ draws.
Because the total number of logical tokens in the entire tree and the
maximum total budget never exceed $N$, these counters are bounded by
$2\left\lceil\sqrt{n\log n}\right\rceil$, and require at most $O(\log n)$ bits.

\subsection{The Algorithm}
We first define the structure of the two primary message types utilized during
the tree expansion and verification phases
before presenting Algorithm~\ref{alg:tree_election}.
\begin{itemize}
    \item \textbf{EXPLORE Message:} Formatted as $\langle \text{EXPLORE}, r_{max}, bundle, budget \rangle$.
    \begin{itemize}
        \item $r_{max}$: The highest candidate rank observed by the sending node.
        %Lower ranks are silently dropped.
        \item $bundle$: The number of independent logical tokens traversing the exact same physical edge in this round, bundled together to conserve bandwidth.
        \item $budget$: The specific portion of the leaf-allocation budget assigned to these bundled tokens,
        dictating how many leaves at the last level they are permitted to spawn.
    \end{itemize}

    \item \textbf{ACK Message:} Formatted as $\langle \text{ACK}, r_{max}, weight \rangle$.
    \begin{itemize}
        \item $r_{max}$: The dominant candidate rank being acknowledged.
        \item $weight$: The aggregated sum of valid logical nodes, including the sender's own logical presence and the returning acknowledgments from its children.
    \end{itemize}
\end{itemize}
\textbf{The Dual Role of the $explore\_cnt$ Array.}
We highlight how the $explore\_cnt$ array shifts its purpose between the algorithm's primary phases.
During \textit{Phase 2}, $explore\_cnt[i]$ acts strictly as a collision counter. It records the total number of highest-ranked logical tokens that converged on the physical node during round $i$.
During \textit{Phase 3}, this array transitions into a topological weight accumulator.
The node initially treats its stored $explore\_cnt[i]$ value as its own base weight at round $i$.
As ACK messages return at round $i$ from its neighbors, the node aggregates their reported weights
directly into its $explore\_cnt[i]$.
This effectively rolls up the total volume of the node's subtree at round $i$.
%ensuring the exact target footprint $N$ is accurately reconstructed and passed back to the root.
%
\begin{algorithm}[t!]
\caption{Tree-Intersection Leader Election in Complete Graphs}
\label{alg:tree_election}
\footnotesize
\begin{algorithmic}[1]
\Statex \textbf{The following code is executed by each node $v \in V$}
\Statex \textbf{Constants (known to all nodes):}
\Statex $G=(V, E)$ where $|V|=n$ \Comment $G$ is a complete graph
\Statex $N = 2\times\lceil \sqrt{n \log n} \rceil$ \Comment{samples volume used for Birthday Paradox}
\Statex $\ell = \text{Branching factor ($\ell \ge 1$) }$ \Comment{$\ell = 1$ for Walk, $\ell \ge 2$ for Tree}
\Statex $T(\ell,N) =$ maximum integer of size $< N$, for which there is a perfect $\ell$-ary tree
\Comment{core tree}
\Statex $H =
\begin{cases}
N-2 & \text{if } \ell = 1 \\
\left\lceil \log_\ell((\ell-1)N+1) \right\rceil -2 & \text{if } \ell \ge 2
\end{cases}$ \Comment{height of the \emph{core tree} with volume $T(\ell,N) < N$}
%\Statex $W_{target} = $ \# of leaves of an almost complete tree of size $N$
%\Comment{total \# of ACKs to verify success}
%\Statex
\Statex \textbf{State variables at node $v$:} \Comment{$H+1$ (not $H$) is the height of the $\ell$-ary tree}
\Statex $r_{v} \gets 0, r_{max} \gets 0, parent[1..H+1] \gets \bot$
\Comment{$r_{v} =$ rank of $v$, $r_{max} =$ dominant rank discovered}
\Statex $explore\_cnt[0..H+1] \gets 0$ \Comment{for each round, counts \# of highest rank incoming EXPLORE tokens}
\Statex $bundle\_cnt \gets 0$ \Comment{a branch may be selected more than once during draws}
\Statex $budget\_sum \gets 0$  \Comment{\# of leaves to be spawn at the last level}
%\Statex $weight\_sum \gets 0$  \Comment{aggregates weights returned by children}
\Statex $B\gets \emptyset$; $M\gets \emptyset$; $A\gets \emptyset$
        \Comment{$B$ multiset of edges, $M$ \& $A$ sets of EXPLORE msgs, ACK msgs}
%\Statex
\Statex \textbf{Phase 1: Candidate Selection (Round 0)}
\State Node $v$ independently becomes a \textit{candidate} with probability $p = \frac{\log n}{n}$
\If{$v$ is a candidate}
    \State $r_v\gets$ generate a random rank from $\{1, \dots, n^4\}$ \Comment{unique rank w.h.p.}
    \State $r_{max} \gets r_{v}$ ; $explore\_cnt[0]\gets 1$; $budget\_sum \gets N- T(\ell,N)$
           \Comment{ready to send tokens} %$\langle \text{EXPLORE}, r_{v}, 1, budget\_sum\rangle$}
\EndIf
%\Statex
\Statex \textbf{Phase 2: Forward Tree-Expansion (Rounds 1 to $H+1$)}
\Comment{$H+1 =$ height of the $\ell$-ary tree}
\For{round $i = 1$ \textbf{to} $H$}\Comment{for each round $i$, first send messages and then receive}
    \If{$explore\_cnt[i-1] \neq 0$}\Comment{at least one token with $r_{max}$ received in round $i-1$}
        \State $B \gets$ Select $\ell \times explore\_cnt[i-1]$ branches to neighbors u.a.r.\ with replacement
        \For{each branch $b\in \operatorname{supp}(B)$}
        \Comment{sequentially; $\operatorname{supp}(B) =$ the set of unique elements in $B$}
             \State $bundle\_cnt \gets$ \# of times $b$ appears in $B$ \Comment{bundle $bundle\_cnt$ tokens in one message}
             \If {$budget\_sum \geq bundle\_cnt\cdot \ell^{H-i+1}$}\Comment{partitions the total budget sequentially}
                 \State $budget\_temp \gets bundle\_cnt\cdot \ell^{H-i+1}$\Comment{assigned the maximum capacity}
             \Else
                 \State $budget\_temp \gets budget\_sum$\Comment{receives remainder budget}
             \EndIf
             \State $budget\_sum \gets budget\_sum - budget\_temp$ \Comment{remanning budget}
             \State Send $\langle \text{EXPLORE}, r_{max}, bundle\_cnt, budget\_temp\rangle$ via branch $b$ \Comment{EXPLORE via $b$}
        \EndFor
 %       \State $explore\_cnt \gets 0$\Comment{set to 0, the \# of incoming explorations, before counting}
    \EndIf

    \State $M \gets$ EXPLORE tokens received in this round \Comment{$budget\_sum = 0$ at this point}
    \For{each $\langle \text{EXPLORE}, r, bundle, budget \rangle \in M$}\Comment{counts \emph{sequentially} \# of high rank tokens}
        \If{$r > r_{max}$} \Comment{new dominant rank discovered}
            \State $r_{max} \gets r; parent[i] \gets \text{arrival edge from sender}$ \Comment{facilitate trace-back}
            \State $explore\_cnt[i] \gets bundle$; $budget\_sum \gets budget$
            \Comment{first high rank incoming exploration}
            \State $explore\_cnt[0..i-1] \gets 0$ \Comment{clear history}
        \ElsIf{$r = r_{max}$ \textbf{and} $parent[i]=\bot$} \Comment{loop or late arrival with same high rank}
            \State $parent[i] \gets \text{arrival edge from sender}$ \Comment{facilitate trace-back}
            \State $explore\_cnt[i] \gets bundle$; $budget\_sum \gets budget$
            \Comment{first high rank exploration in round $i$}
        \ElsIf{$r = r_{max}$ \textbf{and} $parent[i]\neq\bot$} \Comment{internal collision}
            \State $explore\_cnt[i] \gets explore\_cnt[i] + bundle$;
             $budget\_sum \gets budget\_sum + budget$
            \Else
        ~\textbf{drop} $\langle \text{EXPLORE}, r, bundle, budget\rangle$ \Comment{disqualified by a higher rank that arrived earlier}
        \EndIf
    \EndFor
\EndFor

\Statex \textbf{Phase 2: Round $H+1$}
    \If{$budget\_sum \neq 0$}\Comment{ready to spawn $budget\_sum$ new leaves}
        \State $B \gets$ Select $budget\_sum$ branches to neighbors u.a.r.\  with replacement
        \For{each branch $b\in \operatorname{supp}(B)$}\Comment{branches to spawn}
             \State $bundle\_cnt \gets$ \# of times $b$ appears in $B$ \Comment{bundle $bundle\_cnt$ tokens in one message}
             \State Send $\langle \text{EXPLORE}, r_{max}, bundle\_cnt, 0\rangle$ via branch $b$
        \EndFor
 %       \State $explore\_cnt \gets 0$\Comment{set to 0, the \# of incoming explorations, before counting}
    \EndIf

    \State $M \gets$ EXPLORE tokens received in this round \Comment{$M$ is a set of tokens}
    \For{each $\langle \text{EXPLORE}, r, bundle, 0 \rangle \in M$}\Comment{counts \emph{sequentially} \# of high rank tokens}
        \If{$r > r_{max}$} \Comment{new dominant rank discovered}
            \State $r_{max} \gets r; parent[H+1] \gets \text{arrival edge from sender}$ \Comment{facilitate trace-back}
            \State $explore\_cnt[H+1] \gets bundle$;  %\Comment{first high rank incoming exploration}
                   $explore\_cnt[0..H] \gets 0$ \Comment{update state and clear history}
        \ElsIf{$r = r_{max}$ \textbf{and} $parent[H+1]=\bot$} \Comment{loop or late arrival with same high rank}
            \State $parent[i] \gets \text{arrival edge from sender}$ \Comment{facilitate trace-back}
            \State $explore\_cnt[H+1] \gets bundle$  \Comment{first high rank incoming exploration in round $i$}
        \ElsIf{$r = r_{max}$ \textbf{and} $parent[H+1]\neq\bot$} \Comment{internal collision}
            \State $explore\_cnt[H+1] \gets explore\_cnt[H+1] + bundle$ \Comment{adds $bundle$ of high rank explorations}
            \Else
        ~\textbf{drop} $\langle \text{EXPLORE}, r, bundle, 0\rangle$ \Comment{disqualified by a higher rank that arrived earlier}
        \EndIf
    \EndFor
%\Statex \center{\textbf{The algorithm continues on the next page}}
% Before closing the first block, save the current line number state
\makeatletter
\global\edef\savedalgorithmicline{\the\value{ALG@line}}
\makeatother
\end{algorithmic}
\end{algorithm}
% STEP 1: Safely clone the macro by value (breaks the infinite recursion loop)
\makeatletter
\let\originaltheHalgorithm\theHalgorithm
\renewcommand{\theHalgorithm}{\originaltheHalgorithm-cont}
\makeatother
\begin{algorithm}[t!] % Forces the continuation to the top of the next page
\footnotesize
\addtocounter{algorithm}{-1} % Roll back the counter so the number stays identical
\caption{Tree-Intersection Leader Election in Complete Graphs (Continued)}
\begin{algorithmic}[1]
%
% Immediately restore the line counter right after opening the environment
\makeatletter
\setcounter{ALG@line}{\savedalgorithmicline}
\makeatother
\Statex \textbf{Phase 3: Trace-back and Verification (Rounds $H+2$ to $2H+2$)}
\For{round $i = 1$ \textbf{to} $H+1$} \Comment{global rounds $H+2$ to $2H+2$}
    \If{$explore\_cnt[H-i+2] \neq 0$}\Comment{if there is an ACK to send}
        \State Send $\langle \text{ACK}, r_{max}, explore\_cnt[H-i+2] \rangle$ via edge $parent[H-i+2]$
    \EndIf

    \State $A \gets$ ACKs received in this round
    \For{each $\langle \text{ACK}, r, weight \rangle \in A$} \Comment{aggregates weights returned by children + own weight}
        \If{$r = r_{max}$} $explore\_cnt[H-i+1] \gets explore\_cnt[H-i+1] + weight$
        \Else
        ~\textbf{drop} $\langle \text{ACK}, r, weight \rangle$ \Comment{disqualified by a higher rank that arrived later}
        \EndIf
    \EndFor
\EndFor

\Statex \textbf{Phase 4: Termination  (Round $2H+3$)}
\If{$v$ is a candidate \textbf{and} $explore\_cnt[0] = N$} \Comment{$r_{v} = r_{max}$}
     \State $status \gets$ \textbf{LEADER}
     \textbf{else} $status \gets$ NOT\_LEADER
\EndIf
\end{algorithmic}
\end{algorithm}
%
% STEP 3: Safely restore the original macro definition
\makeatletter
\let\theHalgorithm\originaltheHalgorithm
\makeatother

%\newpage
\subsection{Correctness and Analysis of Algorithm~\ref{alg:tree_election}}
We first formalize
two invariants of the algorithm: the \emph{timing invariant} that proves when the ACK
messages arrive, and the \emph{conservation of weight invariant} regarding
the exact preservation of logical weight despite internal collisions.

\begin{lemma}[Timing Invariant of the Silent Pulse] \label{lem:timing-invariant}
For any node $v$ that adopts a candidate's rank $r_{max}$ at round $d$ ($1 \le d \le H+1$) of Phase 2,
$v$ receives all valid ACK tokens from its spawned children precisely in round
$H+1 - d$ of Phase 3. The absence of an ACK in this round safely evaluates to an implicit weight of zero.
\end{lemma}

\begin{proof}
During Phase 2, if $v$ receives and adopts the maximum rank token at round $d$,
it records the arrival edge as $parent[d]$ and forwards EXPLORE tokens to its chosen neighbors at round $d+1$.
A neighbor $u$ that successfully receives this token will record the arrival edge as $parent[d+1]$.
During Phase 3, which executes for rounds $i = 1$ to $H+1$,
a node sends its accumulated weight
via the edge $parent[H-i+2]$. Node $u$ will therefore transmit its ACK back to $v$ when $H-i+2 = d+1$.
Because the network is strictly synchronous,
$v$ evaluates and aggregates all incoming ACKs exactly in round $i = H-d+1$.

Because $v$ expects and processes ACKs from its children strictly in round $H - d+1$,
any branch that was pruned (due to encountering a higher rank) or absorbed
(due to an internal collision) naturally drops out of the communication flow.
By treating non-arriving messages as $0$ weight, the algorithm eliminates the
need for explicit rejection messages while strictly preserving
the layer-by-layer temporal structure of the tree.
\end{proof}

\begin{lemma}[Conservation of Weight] \label{lem:conservation}
Let $v_{max}$ be the candidate with the unique global maximum rank $rank_{max}$.
In Phases 1\&2
(1) the sum of $explore\_cnt[i]$ across all nodes at the end of round
$0 \le i \le H$ is exactly $\ell^i$, and
(2) the sum of $explore\_cnt[H+1]$ across all nodes at the end of round $H+1$
is exactly $N- T(\ell,N)$ (i.e., \# of leaves at the last level)
Consequently, in Phase 3, the total weight returned to $v_{max}$ in round $2H+2$
equals $N$ (i.e., the total number of
EXPLORE tokens with the highest rank submitted plus one).
\end{lemma}
\begin{proof}
For $0 \le i \le H$,
we proceed by induction on the round number $i$ in Phases 1\&2.\\
\textbf{Base case ($i=0$):} Candidate $v_{max}$ initializes $explore\_cnt = 1$, which trivially equals $\ell^0$.\\
\textbf{Inductive step:} Assume the sum of $explore\_cnt$ across all nodes at round $i-1$ is $\ell^{i-1}$. In round $i$
exactly $\ell \times \ell^{i-1} = \ell^i$ EXPLORE tokens with $rank_{max}$ are transmitted in the network. Because no rank in the network is strictly greater than $rank_{max}$, no node drops these tokens.
If multiple tokens arrive at the same node, this node aggregates them by incrementing
$explore\_cnt[i]$ for each token. Thus, the sum of $explore\_cnt$ at round $i$ is exactly $\ell^i$.
As for round $H+1$, exactly $N- T(\ell,N)$ EXPLORE tokens with $rank_{max}$ are transmitted, and
each node upon receiving one of these token will increment by one $explore\_cnt[H+1]$.

In Phase 3, governed by the layer-by-layer retracing established in Lemma~\ref{lem:timing-invariant},
this aggregated weight is passed back without loss.
The precise timing ensures no race conditions occur during weight aggregation,
ensuring the sum of all weights returned to $v_{max}$ equals the number of
EXPLORE tokens with $rank_{max}$ submitted plus one (the root)
which is $N$.
\end{proof}
The following technical lemma
%summarizes the properties of the logical tree construction
is used in the proofs of Lemma~\ref{lem:tree-inter} and Lemma~\ref{thm:tree-complex}.
\begin{lemma}[The Almost Complete Tree Construction]\label{lem:exact-volume-physical-budget}
For every $n \ge 2$ and every integer $1 \le \ell \le N$,
the construction described in Section~\ref{sec:logical_tree_construction}
(as implemented in Algorithm~\ref{alg:tree_election}) produces an almost complete
$\ell$-ary logical tree with exactly $N$ logical nodes.
That is, for every depth $0\le d<H$, each logical node of depth $d$
has exactly $\ell$ logical children, and every logical node of depth $H$
has at most $\ell$ logical children.
The number of supplementary leaves (at the last level) added to the core tree, namely at depth $H+1$,
is the initial budget $\LeafNumber(\ell,N)$.
\end{lemma}
%
%For lack of space
A  detailed proof of Lemma~\ref{lem:exact-volume-physical-budget}
appears in Appendix~\ref{app:Almost-Complete-Tree}.

\begin{lemma}[Tree Intersection Probability]\label{lem:tree-inter}
Let $N = 2\left\lceil \sqrt{n\log n}\right\rceil$
and let $\ell$ be the branching parameter such that $1 \le \ell \le N$.
In a complete graph $K_n$, let $T_1$ and $T_2$
be two independent tree expansions generated by Algorithm~\ref{alg:tree_election},
with the corresponding $\ell$-ary logical trees $T_1'$ and $T_2'$, respectively,
each of exact volume $N$. The probability that $T_1\cap T_2=\emptyset$ is at most $n^{-4}$.
\end{lemma}

\input{Ido-4-alg-2-intersection}

\begin{theorem}[Safety] \label{thm:tree-safety}
With high probability, algorithm~\ref{alg:tree_election} elects at most one leader.
\end{theorem}
\begin{proof}
Let $v_{max}$ be the candidate with the unique maximum rank $rank_{max}$. Suppose for contradiction that another candidate
$u \neq v_{max}$ with rank $r_u < rank_{max}$ is elected.
For $u$ to be elected, it must evaluate $explore\_cnt[0] =N$ at round $2H+3$.
By Lemma~\ref{lem:tree-inter}, $u$'s tree expansion and $v_{max}$'s tree expansion intersect at some node $z$ with high probability. Because expansions proceed synchronously layer-by-layer, node $z$ updates its local $r_{max}$ to $rank_{max}$.
(Any subsequent or concurrent EXPLORE or ACK token carrying $r_u$ arriving at $z$ will trigger the ``drop'' condition.)
Thus, node $z$ will never send an ACK message to $u$.
By Lemma~\ref{lem:conservation},
the lost weight from this unsent ACK by $z$ is never returned to $u$, strictly bounding $u$'s final
$explore\_cnt[0] < N$. Thus, $u$ cannot be elected.
\end{proof}

\begin{theorem}[Liveness] \label{thm:tree-liveness}
With high probability, exactly one leader is elected.
\end{theorem}
\begin{proof}
The candidate $v_{max}$ possesses the global maximum rank. Therefore, its tokens can never trigger the drop condition at any node.
By Lemma~\ref{lem:conservation}, the ``silent pulse'' aggregation successfully prevents weight loss during internal collisions. Consequently, $v_{max}$ receives exactly $N$ acknowledgments (by weight) in Phase 3,
fulfilling the termination condition. Combined with Theorem~\ref{thm:tree-safety}, $v_{max}$ is the unique leader w.h.p.
\end{proof}

\begin{theorem}[\textsf{CONGEST}] \label{thm:tree-congest}
The constraint of the \textsf{CONGEST} model is satisfied.
\end{theorem}
\begin{proof}
%The CONGEST model requires that each edge transmits at most $O(\log n)$ bits per round.
Algorithm~\ref{alg:tree_election} utilizes two types of messages:
EXPLORE and ACK, which carries a rank value.
Because candidate ranks are chosen from the domain $\{1, \dots, n^4\}$, representing a rank requires exactly $\lceil \log_2(n^4) \rceil = 4 \log_2 n$ bits, which is $O(\log n)$. Because every node only forwards the single maximum rank it has seen, an edge carries at most one token per round.
In Phase 2\&3,
Explore and ACK has the formats
$\langle \text{EXPLORE}, r_{max}, bundle, budget\rangle$, and
$\langle \text{ACK}, r_{max}, weight \rangle$.
The rank $r_{max}$ again takes $O(\log n)$ bits. The maximum possible values for the
$bundle\_cnt, budget\_temp$ and $explore\_cnt[i]$ variables
is bounded by $N$.
Since $N \approx 2\sqrt{n \log n}$,
the number of bits required to represent these variables is at most
$\lceil \log_2(2\sqrt{n \log n}) \rceil \approx \frac{1}{2}\log_2 n + \frac{1}{2}\log_2(\log n) + 1$.
This is strictly bounded by $O(\log n)$ bits.

Since both the EXPLORE and ACK tokens require $O(\log n)$ bits to encode,
and the pruning mechanism ensures at most one such token is sent across any specific edge in a single round,
the algorithm strictly adheres to the \textsf{CONGEST} bandwidth constraints.
\end{proof}

\begin{theorem}[Complexity Bounds] \label{thm:tree-complex}
With high probability, Algorithm~\ref{alg:tree_election} achieves:
\begin{enumerate}
    \item \textbf{Time Complexity:}
     $O(\sqrt{n \log n})$
     rounds if
     $\ell = 1$,
     and
     $O(\log_\ell \sqrt{n \log n})$
     rounds if
     $\ell \ge 2$.
    \item \textbf{Total Message Complexity:}
    $O(\sqrt{n} \log^{1.5} n)$
    messages.
    \item \textbf{Per-Node Message Complexity:}
    $O(\ell)$ messages.
\end{enumerate}
\end{theorem}
\input{Ido-5-alg-2-complexity.tex}

\section{Related Work}
\label{sec:related_work}

Symmetry breaking and leader election (LE) are among
the most extensively studied problems in distributed computing.
A critical nuance in the history of the LE problem is the distinction between
\textit{explicit} and \textit{implicit} LE. In the explicit version,
every node must terminate knowing the identifier of the elected leader.
This requirement naturally leads to a message complexity
of $\Omega(n)$, as the leader's identity must be disseminated to every node.

Early research focused heavily on solving deterministically the explicit version in the asynchronous model. The seminal work of Gallager, Humblet, and Spira (GHS) \cite{GHS83} established a foundational baseline for general graphs, achieving $O(m + n \log n)$ total messages and $O(n \log n)$ time for minimum spanning tree construction and explicit LE in general graphs,
where $m$ is the number of edges.
The $\Omega(n \log n)$ term represents the inescapable cost of deterministic symmetry breaking,
as first demonstrated on the ring topology \cite{Burns80, Pachl82}.
The lower bounds $\Omega(m)$ (and $\Omega(D)$ for time complexity, where $D$ is the diameter), were proven in \cite{KuttenJACM2015},
and hold also for randomized algorithms in the synchronous model.

To achieve sublinear message complexity, two shifts were necessary:
the introduction of randomization and the move toward \textit{implicit} LE.
Recent advancements in the study of implicit LE have focused
on optimizing total message complexity and execution time.
Kutten et al. \cite{Kutten2015} introduced a synchronous randomized algorithm with
sublinear total message complexity, and
an almost matching lower bound
showing that $\Omega(\sqrt{n})$ messages are needed for any randomized LE
algorithm that succeeds with probability at least $1/e + \epsilon$, for any small constant $\epsilon >0$.
The approach of \cite{Kutten2015} inherently requires high local
bandwidth (a massive concurrent broadcast).
Our work is the first to parameterize and optimize the per-node
message complexity while maintaining strictly sublinear total communication.

In \cite{GRS2018}, a randomized synchronous algorithm is presented
that solves implicit LE in a general network $G$ with $O(t_{\mathrm{mix}}(G) \cdot \sqrt{n} \log^{3.5} n)$
messages and in $O(t_{\mathrm{mix}}(G) \cdot \log^{2} n)$ time, but has high local
per-node message complexity.
Correspondingly, a nearly matching lower bound is presented, showing that
$\Omega(\sqrt{n})/{(\phi)^{3/4}}$ messages are needed for any implicit LE algorithm that succeeds
with probability at least $1 - o(1)$,
where $\phi$ refers to the conductance of a graph.
A key difference between \cite{GRS2018} and \cite{Kutten2015} is that in \cite{GRS2018}
it is \emph{not} assumed that the mixing time is known, which significantly
complicates the problem.

In \cite{KM2021},
a randomized synchronous algorithm is presented
that solves LE in
general network $G$ with
$O(\min\{\sqrt{n\cdot t_{\mathrm{mix}}(G)/ \phi}, n/(\phi \log n)\})$ messages
and $O(t_{\mathrm{mix}}(G)~\log n + \phi^{-1} \log^3 n)$ time
assuming the conductance and the mixing time are known.

As for asynchronous complete graphs, a randomized LE algorithm that achieves
sublinear $\sqrt{n} \log^{1.5} n$  message complexity w.h.p. and $O(\log n)$
time complexity, but has high local
per-node bandwidth, was presented in \cite{EKRT2025}.

%Per-node message complexity was investigated in the context of the consensus problem.
%In \cite{AAKS2018}, a randomized consensus is proposed with expected message complexity
%$O( n^2 \log ^2 n )$ against a strong (adaptive) adversary,
%in an asynchronous message-passing model in which less than $n / 2$ processes may fail by crashing.
%The algorithm is also \emph{locally-efficient}, ensuring that no process sends or receives
%more than $O( n \log ^3 n )$ expected messages.

%\section{Discussion and Future Work}
\section{Discussion}
\label{sec:discussion}

\textbf{Leader election in general graphs.}
In Section~\ref{sec:Tree-Intersection:CompleteGraphs},
we established the tree-intersection strategy for complete graphs,
where a single random hop from any node results in a perfectly uniform sample from the entire network.
In a general graph, this property no longer holds; a single hop moves a message only to a direct neighbor,
keeping the exploration strictly localized. It is tempting to adapt Algorithm~\ref{alg:tree_election}
by introducing a \emph{logical-to-physical mapping} of the tree structure.
The core idea is to treat every logical edge
%in the $\ell$-ary tree from Algorithm~\ref{alg:tree_election}
as a physical random walk of length $t_{\mathrm{mix}}(G)$ -- the mixing time of $G$.%
\footnote{The mixing time, $t_{\mathrm{mix}}(G)$, of an $n$-node graph $G$, is defined as the number of steps required for the random walk distribution to get $1/(2n)$-close to the stationary distribution $\pi$, with respect to the maximum norm, regardless of the starting node.}
In such a \emph{stretched tree}, a candidate initiates $\ell$ independent random walks.
Only after $t_{\mathrm{mix}}(G)$ physical rounds do the endpoints of these walks act as ``logical children''
and branch out into another $\ell$ random walks each.

While conceptually elegant, formalizing this intuition is highly non-trivial.
Unlike the complete graph where the stationary distribution is uniform,
random walks in general graphs converge to a distribution proportional
to node degrees ($\pi(v) = \frac{d(v)}{2m}$).
This skew means that dense areas of the graph are sampled disproportionately,
which heavily alters the intersection dynamics and complicates the Birthday Paradox bounds.
Furthermore, one must carefully untangle the probabilistic correlations introduced by
the branching nature of the walks.

Such a transformation decouples the graph's topology from the intersection probability.
Since we replace every logical round with $t_{\mathrm{mix}}(G)$ physical rounds,
the time complexity scales to $O(t_{\mathrm{mix}}(G) \cdot \log_\ell \sqrt{n \log n})$ for $\ell \ge 2$, and
the total message complexity scales to $O(t_{\mathrm{mix}}(G) \cdot \sqrt{n} \log^{1.5} n)$.

However, while this transformation successfully achieves sublinear total message complexity,
the per-node message complexity would degrade, potentially scaling as high as the total message complexity itself.
This theoretical wall arises inevitably in graphs with severe structural bottlenecks.
Consider a double star graph: even if we simulate logical tree expansions using random walks,
one of the two center nodes lies on every path of length greater than one.
Therefore, if the algorithm requires a total message complexity of $M$
to ensure tree intersection, the center nodes are forced to process $\Omega(M)$ messages.
In such topologies, the maximum per-node message complexity
binds to the total message complexity, rendering our algorithmic dial, $\ell$,
entirely ineffective at balancing the local communication load.\\
\\
\textbf{Our central contribution.}
We presented a generalized framework for randomized implicit leader election
in synchronous complete graphs under the \textsf{CONGEST} model.
Our central contribution is the branching factor $\ell$, which allows system designers
to smoothly navigate the trade-off between time complexity ($O(\log_\ell \sqrt{n \log n})$ for $\ell \ge 2$)
and per-node message complexity ($O(\ell)$),
while maintaining a strictly sublinear total message complexity.
This framework unifies and generalizes previous state-of-the-art bounds,
demonstrating that single-path intersections and massive concurrent broadcasts
are merely two extremes of the same algorithmic spectrum.\\
\\
\textbf{Per-node complexity vs.\ per-edge complexity.}
As motivated in the Introduction, a natural consequence of our framework is its inherent
guarantee regarding total per-edge message complexity. While the \textsf{CONGEST}
model naturally prevents \textit{per-round} link congestion,
our total per-node bound provides a structural guarantee over the
lifetime of the algorithm. For instance, an algorithm where a high-degree node sends
exactly one message per incident edge keeps the total per-edge complexity minimal,
yet allows the total per-node complexity to scale prohibitively with the node's degree.
By optimizing the per-node metric instead, our framework ensures that no single
node is exhausted and, as a direct corollary, no single communication link is
overwhelmed during the entire algorithm.\\
\\
\textbf{Open questions.}
Several compelling open questions remain for future investigation:
\begin{itemize}
\item \textit{Beyond Complete Graphs:} Establishing a parameterized trade-off between time and per-node
      message complexity for broader graph classes, such as regular graphs,
      remains a compelling direction for future research.
\item \textit{Parameterized Lower Bounds:} An immediate question is whether the trade-off established by $\ell$ is strictly optimal. Deriving a matching parameterized lower bound would prove that this latency-bandwidth exchange is an inherent property of symmetry breaking.
\item \textit{Fault Tolerance:} Extending our algorithm to be robust against node crashes or link failures
    would be valuable.
\item \textit{Explicit Leader Election:} Coupling our symmetry-breaking mechanism, which already maintains
low per-node complexity, with an efficient broadcast protocol could yield highly optimized results
for the explicit variant of the problem.
%\item \textit{Per-node message complexity:}
%Establishing tight upper and lower bounds on the per-node message complexity
%for more general topologies would be interesting.
\end{itemize}

\appendix
\section{Proof of Lemma~\ref{lem:exact-volume-physical-budget}}
\label{app:Almost-Complete-Tree}

\input{Ido-3-alg-2-logical-tree-analysis.tex}

\newpage
\bibliography{election.bib}

%\newpage
%\input{clique}

\end{document}

%% file: Ido-0-alg-1-intersection.tex
\begin{proof}
Let
$X_1,\dots,X_N$
and
$Y_1,\dots,Y_N$
be the
$N = L + 1 =2\left\lceil \sqrt{n\log n}\right\rceil$
nodes visited by
$W_1$
and
$W_2$,
respectively, in the order in which they are visited, where
$X_1$
and
$Y_1$
are the starting nodes of the walks.
Define
$S_X:=\{X_1,\dots,X_N\}$
and
$D:=|S_X|$.
The event
$W_1\cap W_2=\emptyset$
is equivalent to
$\forall j\in [N]:Y_j\notin S_X$.
Hence, by the chain rule, for every realization
$S_X=S$,
\[
\Pr[W_1\cap W_2=\emptyset\mid S_X=S]
=
\prod_{j=1}^N
\Pr\!\left[
Y_j\notin S
\;\middle|\;
Y_1,\dots,Y_{j-1}\notin S,\ S_X=S
\right].
\]
Let us bound each term in the product.
For
$j=1$,
we use the trivial upper bound
$1$.
For every
$j\ge 2$,
under the above conditioning the current step
$Y_{j-1}$
of
$W_2$
lies outside
$S$.
Let
$d=|S|$.
Since
$S$
contains
$d$
distinct nodes, the events of moving to each of these nodes in the next step
$Y_j$
are disjoint.
Thus, the probability that the walk enters
$S$
in the next step is exactly
$d/(n-1)$.
Therefore,
\[
\Pr[W_1\cap W_2=\emptyset\mid S_X=S]
\le
1\cdot
\left(1-\frac{d}{n-1}\right)^{N-1}
\le
\exp\left(-\frac{d(N-1)}{n-1}\right).
\]
Averaging the preceding bound over all realizations
$S_X=S$
with
$|S|=d$,
we get
\[
\Pr[W_1\cap W_2=\emptyset\mid D=d]
\le
\exp\left(-\frac{d(N-1)}{n-1}\right).
\]
It remains to lower-bound
$D$.
For
$t\in[N]$,
let
$I_t=\mathbf{1}\!\left[X_t\notin\{X_1,\dots,X_{t-1}\}\right]$.
Then
$D=\sum_{t=1}^N I_t$.
Reveal
$X_1,\dots,X_N$
sequentially.
For
$t=1$,
we have
$I_1=1$.
For
$t\ge 2$,
if by step
$t-1$
the walk has visited
$k$
distinct nodes, then
$X_{t-1}$
is one of these
$k$
nodes, and among the
$n-1$
possible choices for
$X_t$,
at most
$k-1$
are old.
Since
$k\le t-1\le N = o(n)$,
for all sufficiently large
$n$,
\[
\Pr[I_t=1\mid X_1,\dots,X_{t-1}]
\ge
1-\frac{k-1}{n-1}
\ge
\frac{7}{8}.
\]
Hence
$D$
stochastically dominates
$\operatorname{Bin}(N,7/8)$.\footnote{This is the standard coupling: if
$q_t=\Pr[I_t=1\mid X_1,\dots,X_{t-1}]$,
then
$q_t\ge 7/8$.
Using independent
$U_t\sim\operatorname{Unif}[0,1]$,
one may realize
$I_t=\mathbf{1}[U_t\le q_t]$,
while
$B_t=\mathbf{1}[U_t\le 7/8]$
are independent Bernoulli
$(7/8)$
variables and satisfy
$B_t\le I_t$
for all
$t$.}
By a Chernoff bound,
\[
\Pr[D<3N/4]
\le
\Pr[\operatorname{Bin}(N,7/8)<3N/4]
\le
\exp\left(-\frac{(1/7)^2\cdot (7N/8)}{2}\right)
\le
e^{-\Omega(N)}.
\]
Using the bound conditioned on
$D=d$
and splitting according to whether
$D<3N/4$,
%we obtain
\[
\Pr[W_1\cap W_2=\emptyset]
\le
e^{-\Omega(N)}
+
\exp\left(-\frac{3N(N-1)}{4(n-1)}\right).
\]
Since
$N=2\left\lceil \sqrt{n\log n}\right\rceil$,
\[
\frac{3N(N-1)}{4(n-1)}
\ge
(3-o(1))\log n
=
\left(\frac{3}{\ln 2}-o(1)\right)\ln n.
\]
As
$3/\ln 2>4.3$,
we have
$%\[
\exp\left(-\frac{3N(N-1)}{4(n-1)}\right)
\le
n^{-4.1}
=
o(n^{-4}).
$\\ %\]
Also,
$e^{-\Omega(N)}=o(n^{-4})$.
Thus both terms above are
$o(n^{-4})$,
and in particular, for all sufficiently large
$n$,
their sum is at most
$n^{-4}$.
Therefore,
$\Pr[W_1\cap W_2=\emptyset]\le n^{-4}$.
\end{proof} 

%% file: Ido-1-alg-1-complexity.tex
\begin{proof}
\textbf{Time:} The algorithm runs for exactly
$2L$
rounds.
Round
$0$
and
$2L+1$
are not counted as they do not involve communication.
With
$L = 2\times \lceil \sqrt{n \log n} \rceil - 1$,
time is as claimed.\\
\\
\textbf{Total Messages:}
Let $C$ be the set of candidates.
Since each node becomes a candidate with probability
$p = \frac{\log n}{n}$,
the expected number of candidates is
$E[|C|] = n \cdot \frac{\log n}{n} = \log n$.
By Chernoff bounds,
$|C| \leq 4 \log n$ w.h.p.%
\footnote{
We use a common form of the Chernoff bound %for the "upper tail":
$P(X \ge (1+\delta)\mu) \le ( {e^\delta}/{(1+\delta)^{(1+\delta)}} )^\mu$,
where: $\mu = \log n$, and $\delta = 3$.
Thus,
$P(|C| \ge 4 \log n) \le ({e^3}/{4^4})^{\log n} \le (0.078)^{\log n} = n^{\log(0.078)} \approx n^{-3.68}$
}
Each candidate generates exactly one token that travels at most
$L$
steps forward and
$L$
steps back.
The total number of messages is at most
$|C| \cdot 2 L$.
Substituting
$L = 2\times \lceil \sqrt{n \log n} \rceil - 1$,
w.h.p.:
$M \leq (4 \log n) \cdot 4 \lceil \sqrt{n \log n} \rceil = O(\sqrt{n} \log^{1.5} n).$\\
\\
\textbf{Per-Node Complexity:}
Fix a node $v$.
We bound the number of forward messages received by $v$.
Let
$X_v^F$
denote this number.
Let
$\mathcal{E}$
be the event that
$|C|\le 4\log n$.
As shown above,
$\Pr[\neg \mathcal{E}]\le n^{-3.67}$.
Condition on an arbitrary candidate set
$C$
such that
$|C|\le 4\log n$.
Order all possible forward trials in some fixed order, for example first by the ID of the candidate that generated the token, and then by the step number of the token.
Let
$N = L + 1 = 2 \lceil \sqrt{n \log n} \rceil$.
There are at most
$4N\log n$
such trials.
For the $i$th trial, let
$I_i$
be the indicator for the event that the corresponding forward message is received by
$v$.
Then
$X_v^F=\sum_{i=1}^{4N\log n} I_i$.

The indicators
$I_i$
are not necessarily independent.
However, for every history of the process before the
$i$th
trial, the probability that the next forward message is received by
$v$
is at most
$1/(n-1)$.
Indeed, if the token is currently at $v$, then the probability that the next message is received by
$v$
is
$0$;
otherwise, the next node is chosen uniformly among the
$n-1$
neighbors.
Thus, conditioned on
$\mathcal{E}$,
random variable
$X_v^F$
is stochastically dominated by the random variable
$Y \sim \operatorname{Bin}(4N\log n,1/(n-1))$.%
\footnote{%
This is the standard coupling under the conditioning on
$\mathcal{E}$:
if
$q_i=\Pr[I_i=1\mid \mathcal{E},\mathcal{H}_{i-1}]$,
where
$\mathcal{H}_{i-1}$
is the full history before the
$i$th trial, then
$q_i\le 1/(n-1)$.
Using independent
$U_i\sim\operatorname{Unif}[0,1]$,
one may realize
$I_i=\mathbf{1}[U_i\le q_i]$,
while
$B_i=\mathbf{1}[U_i\le 1/(n-1)]$
are independent Bernoulli
$(1/(n-1))$
variables and satisfy
$I_i\le B_i$
for all
$i$.}
Therefore, for every fixed integer
$k$,
\[
\Pr[X_v^F\ge k\mid \mathcal{E}]
\le
\Pr[Y\ge k].
\]
If
$Y\ge k$,
then there exists a set
$S$
of
$k$
successful trials.
Therefore, by a union bound over all such sets
$S$,
\[
\Pr[Y\ge k]
\le
\sum_{\substack{S\subseteq [4N\log n]\\ |S|=k}}
\Pr[\text{all trials in } S \text{ are successful}]
=
\binom{4N\log n}{k}
\left(\frac{1}{n-1}\right)^k
\le
\frac{1}{k!}
\left(\frac{4N\log n}{n-1}\right)^k .
\]
Taking
$k=6$
and using
$N=2\lceil\sqrt{n\log n}\rceil$, we get
\[
\Pr[X_v^F\ge 6\mid \mathcal{E}]
=
O\left(\frac{\log^9 n}{n^3}\right).
\]
By a union bound over all nodes,
\[
\Pr[\exists v\in V: X_v^F\ge 6\mid \mathcal{E}]
\le
n\cdot O\left(\frac{\log^9 n}{n^3}\right)
=
O\left(\frac{\log^9 n}{n^2}\right).
\]
Thus,
\[
\Pr[\exists v\in V: X_v^F\ge 6]
\le
\Pr[\neg\mathcal{E}]
+
\Pr[\exists v\in V: X_v^F\ge 6\mid \mathcal{E}]
\le
n^{-3.67}
+
O\left(\frac{\log^9 n}{n^2}\right)
\le
\frac{1}{n},
\]
for all sufficiently large $n$.
Thus, with probability at least $1-1/n$, every node receives at most $5$ forward messages.
It remains to translate this into total processed messages, namely sent and received messages in both the forward and trace-back phases.
During the forward phase, every non-initial forward message sent by $v$ is caused by a previous forward message received by $v$.
In addition, $v$ may send one initial message if it is a candidate.
Hence, the number of forward messages sent by $v$ is at most
$X_v^F+1$.

During the trace-back phase, every trace-back message sent by $v$ corresponds to a forward message previously received by $v$, and every trace-back message received by $v$ corresponds to a forward message previously sent by $v$.
Therefore, the number of trace-back messages sent by $v$ is at most $X_v^F$, and the number of trace-back messages received by $v$ is at most $X_v^F+1$.
Consequently, the total number of messages processed by $v$ is at most
\[
X_v^F+(X_v^F+1)+X_v^F+(X_v^F+1)
=
4X_v^F+2.
\]
Therefore, with probability at least $1-1/n$, every node processes at most $22$ messages.
Thus, the per-node message complexity is $O(1)$.
\end{proof} 

%% file: Ido-4-alg-2-intersection.tex
\begin{proof}
By Lemma~\ref{lem:exact-volume-physical-budget},
both logical trees
$T_1'$
and
$T_2'$
contain exactly
$N$
logical nodes.
Let
$X'_1,\dots,X'_N$
and
$Y'_1,\dots,Y'_N$
be the logical nodes of
$T_1'$
and
$T_2'$,
respectively, ordered in BFS order, where
$X'_1$
and
$Y'_1$
are the roots.
For every
$i\in[N]$,
let
$X_i$
(resp.,
$Y_i$)
be the physical node reached by the token corresponding to
$X'_i$
(resp.,
$Y'_i$).
Then
$T_1$
and
$T_2$
are exactly the sets
$\{X_1,\dots,X_N\}$
and
$\{Y_1,\dots,Y_N\}$,
respectively.

The rest of the proof is the same as the proof of
Lemma~\ref{lem:RW-intersection},
with the following minor difference.
In Lemma~\ref{lem:RW-intersection},
for every
$j\ge 2$,
the node
$Y_j$
is sampled from
$Y_{j-1}$.
Here,
$Y_j$
is sampled from the physical node corresponding to the logical parent of
$Y'_j$,
which appears earlier in the BFS order, but is not necessarily
$Y_{j-1}$.
This makes no difference, since the proof conditions on all previously exposed physical nodes
$Y_1,\dots,Y_{j-1}$
being outside
$S_X=\{X_1,\dots,X_N\}$.
Thus, the physical node corresponding to the logical parent of
$Y'_j$
is outside
$S_X$,
and the same bound applies.

The same observation applies when lower-bounding the number of distinct nodes in
$S_X$.
When exposing
$X_i$
in BFS order, the node
$X_i$
is sampled from the physical node corresponding to the logical parent of
$X'_i$,
which appears earlier in the BFS order.
Thus, conditioning on the previously exposed physical nodes gives the same lower bound on the probability that
$X_i$
is new as in Lemma~\ref{lem:RW-intersection}.
Therefore, the calculation from Lemma~\ref{lem:RW-intersection} applies verbatim and gives
$%\[
\Pr[T_1\cap T_2=\emptyset]
\le
n^{-4},
$%\]
for all sufficiently large
$n$.
\end{proof} 

%% file: Ido-5-alg-2-complexity.tex
\begin{proof}$~$
\textbf{Time:}
Let
$N = 2\left\lceil \sqrt{n\log n}\right\rceil = O(\sqrt{n \log n})$.
The algorithm executes for
$2(H+1)$
rounds.
For
$\ell = 1$,
since
$H = N - 2$,
the time complexity is $O(\sqrt{n \log n})$.
For
$\ell\geq 2$,
since
\[
H
=
\left\lceil
\log_\ell((\ell-1)N+1)
\right\rceil
-2 \le O(\log_\ell (2 \ell N))
\le
O(\log_\ell (2\ell) + \log_\ell (N))
\le
O(\log_\ell (N)),
\]
the time complexity is
$O(\log_\ell \sqrt{n \log n})$.\\
\\
\textbf{Total Messages:}
As proven in Theorem~\ref{thm:RW:Complexity Bounds},
the number of candidates is bounded by
$|C|\le 4\log n$
w.h.p.
By Lemma~\ref{lem:exact-volume-physical-budget},
the logical tree generated by each candidate has volume
$N$.
Thus, each candidate generates at most
$N-1$
EXPLORE messages.
Therefore, the total number of EXPLORE messages over all candidates is at most
$|C|\cdot (N-1)\le |C|\cdot N$.

Each ACK message corresponds to at least one generated EXPLORE message.
Indeed, ACKs are only sent back along edges that were created by EXPLORE messages, and if several ACKs are bundled, this only decreases the number of physical ACK messages.
Hence, Phase 3 generates at most the same number of messages as Phase 2.
Therefore, the total number of messages satisfies
$
M
\le
2|C|N
\le
2\cdot (4\log n)\cdot 2\left\lceil \sqrt{n\log n}\right\rceil
=
O(\sqrt{n}\log^{1.5}n)
$
w.h.p.\\
\\
\textbf{Per-Node Complexity:}
Fix a physical node
$v$.
We first bound the number of logical EXPLORE units received by
$v$.
Let
$X_v^E$
denote this number.
Let
$\mathcal{E}$
be the event that
$|C|\le 4\log n$.
As in Theorem~\ref{thm:RW:Complexity Bounds},
$\Pr[\neg \mathcal{E}]\le n^{-3.67}$.
Condition on an arbitrary candidate set
$C$
such that
$|C|\le 4\log n$.
Order all logical EXPLORE trials in some fixed order, first by the ID of the candidate that generated the logical tree, and then by the BFS order of the logical tree.
By Lemma~\ref{lem:exact-volume-physical-budget},
each candidate generates a logical tree of exact volume
$N$,
and therefore there are at most
$4N\log n$
logical EXPLORE trials.

For the
$i$th trial, let
$I_i$
be the indicator for the event that the corresponding logical EXPLORE unit is received by
$v$.
Then
$X_v^E=\sum_{i=1}^{4N\log n} I_i$.
As in the per-node proof of Theorem~\ref{thm:RW:Complexity Bounds},
for every history before the
$i$th trial, the probability that the next logical EXPLORE unit is received by
$v$
is at most
$1/(n-1)$.
Indeed, if the physical node corresponding to the logical parent is
$v$,
then the probability that the next logical child is mapped to
$v$
is
$0$;
otherwise, the next physical node is chosen uniformly among the
$n-1$
neighbors.

Thus, conditioned on
$\mathcal{E}$,
the random variable
$X_v^E$
is stochastically dominated by the random variable
$\operatorname{Bin}(4N\log n,1/(n-1))$.
The same calculation as in Theorem~\ref{thm:RW:Complexity Bounds}, with
$X_v^E$
in place of
$X_v^F$,
gives
\[
\Pr[\exists v\in V:X_v^E\ge 6]
\le
\frac{1}{n},
\]
for all sufficiently large
$n$.
Thus, with probability at least
$1-1/n$,
every physical node receives at most
$5$
logical EXPLORE units.

It remains to translate this into total processed physical messages.
Physical bundling can only decrease the number of messages, so it is enough to upper-bound the number of logical units processed by
$v$.
During Phase 2, every logical EXPLORE unit received by
$v$
can cause at most
$\ell$
logical EXPLORE units to be sent by
$v$,
because every logical node has at most
$\ell$
children by Lemma~\ref{lem:exact-volume-physical-budget}.
In addition,
$v$
may send at most
$\ell$
initial logical EXPLORE units if it is a candidate.
Hence, the number of EXPLORE units sent by
$v$
is at most
$\ell X_v^E+\ell$.

During Phase 3, every ACK unit sent by
$v$
corresponds to a logical EXPLORE unit previously received by
$v$.
Thus, the number of ACK units sent by
$v$
is at most
$X_v^E$.
Similarly, every ACK unit received by
$v$
corresponds to a logical EXPLORE unit previously sent by
$v$.
Therefore, the number of ACK units received by
$v$
is at most
$\ell X_v^E+\ell$.

Consequently, the total number of logical units processed by
$v$
is at most
\[
X_v^E+(\ell X_v^E+\ell)+X_v^E+(\ell X_v^E+\ell)
=
(2\ell+2)X_v^E+2\ell.
\]
Since, with probability at least
$1-1/n$,
we have
$X_v^E\le 5$
for every node
$v$,
every node processes at most
$
(2\ell+2)\cdot 5+2\ell
=
12\ell+10
=
O(\ell)
$
messages.
Thus, the per-node message complexity is
$O(\ell)$.
\end{proof} 

%% file: Ido-3-alg-2-logical-tree-analysis.tex
Section~\ref{sec:logical_tree_construction}
presents two definitions of
$H$:
an explicit formula and a maximum feasible height.
In Lemma~\ref{lem:exact-core-height},
we use the explicit formula definition and show that it is well defined and equivalent to the maximum definition.

\begin{lemma}\label{lem:exact-core-height}
The value
$H$,
namely
$H=N-2$
if
$\ell=1$,
and
\[
H
=
\left\lceil
\log_\ell((\ell-1)N+1)
\right\rceil
-2
\]
if
$2 \le \ell \le N$,
satisfies the following properties for every
$n\ge 2$:
\begin{enumerate}
    \item
    $H$
    is well defined.
    That is, when
    $2\le \ell\le N$,
    the logarithm is taken with a valid base and over a strictly positive input.

    \item
    $H\ge 0$.

    \item
    The set
    $\{h \in \mathbb{Z}_{\ge 0} \mid T_h<N\}$
    is a nonempty bounded integer set.
    That is, the maximum integer
    $h\ge 0$
    such that
    $T_h<N$
    is well defined.
    Moreover,
    $H$
    is exactly that maximum,
    and
    $T_H < N \le T_{H+1}$.
\end{enumerate}
\end{lemma}

\begin{proof}
Since
$N=2\left\lceil\sqrt{n\log n}\right\rceil$,
for
$n \ge 2$,
we have
$N \ge 2$.
Let
$\mathcal{H} := \{h \in \mathbb{Z}_{\ge 0} \mid T_h<N\}$.
For both
$\ell=1$
and
$\ell\ge 2$,
we have
$T_0=1 < 2 \le N$.
Thus,
$\mathcal{H} \neq \emptyset$.
Since
$T_h \overset{h \to \infty}{\longrightarrow}\infty$,
there exists an integer
$M$
such that for every
$h\ge M$,
we have
$T_h\ge N$.
Therefore, every
$h$
satisfying
$T_h<N$
must satisfy
$h<M$,
so the set
$\mathcal{H}$
is bounded.
Since it is also nonempty, it has a maximum.
Therefore,
$\max \mathcal{H}$,
(i.e.,
the maximum integer
$h\ge 0$
such that
$T_h<N$)
is well defined.

First, suppose that
$\ell=1$.
Then
$H=N-2$
is well defined.
Also, since
$N\ge 2$,
we have
$H\ge 0$.
Moreover,
$T_h=h+1$.
Thus,
$T_h<N$
is equivalent to
$h+1<N$.
Since
$N$
is an integer, this is equivalent to
$h\le N-2$.
Therefore, the maximum such integer is indeed
$H=N-2$.
Moreover,
\[
T_H = H + 1 = N - 1 < N = H + 2 = T_{H+1},
\]
so
$T_H < N \le T_{H+1}$.

Now suppose that
$2\le \ell\le N$.
Let
\[
A := \log_\ell((\ell-1)N+1).
\]
First,
$A$
is well defined.
Indeed, since
$\ell\ge 2$,
the base of the logarithm is valid.
Moreover, since
$2\le \ell\le N$,
we have
\begin{equation}\label{eq:log-input-positive}
(\ell-1)N+1
=
\ell N-N+1
\ge
2N-N+1
=
N+1
\ge
\ell+1
>
\ell
>
0.
\end{equation}
Thus,
$A=\log_\ell((\ell-1)N+1)$
is well defined.

The proposed value is
$H=\lceil A\rceil-2$.
We show that
$H\ge 0$.
By Inequality~\eqref{eq:log-input-positive},
we have
$(\ell-1)N+1>\ell$.
Therefore, since
$\ell\ge 2$
and
$\log_\ell(x)$
is strictly monotonically increasing in
$x>0$,
we have
$A>\log_\ell(\ell)=1$.
Thus,
$\lceil A\rceil\ge A>1$.
Since
$\lceil A\rceil$
is an integer,
$\lceil A\rceil\ge 2$.
Hence,
$H=\lceil A\rceil-2\ge 0$.

We now show that
$T_H<N$.
Since
$\lceil A\rceil-1<A$,
we have
$H+1=\lceil A\rceil-1<A$.
Therefore,
\[
\ell^{H+1}<\ell^A=(\ell-1)N+1.
\]
It follows that
\[
T_H
=
\frac{\ell^{H+1}-1}{\ell-1}
<
N.
\]

We next show that
$T_{H+1}\ge N$.
Since
$\lceil A\rceil\ge A$,
we have
$H+2=\lceil A\rceil\ge A$.
Therefore,
\[
\ell^{H+2}\ge \ell^A=(\ell-1)N+1.
\]
It follows that
\[
T_{H+1}
=
\frac{\ell^{H+2}-1}{\ell-1}
\ge
N.
\]

Thus,
\[
T_H<N\le T_{H+1}.
\]
Moreover,
$T_h$
is increasing in
$h$:
for
$\ell=1$,
we have
$T_{h+1}-T_h=1>0$,
and for
$\ell\ge 2$,
we have
$T_{h+1}-T_h=\ell^{h+1}>0$.
This proves that
$H$
is exactly the maximum integer
$h \ge 0$
such that
$T_h<N$.
\end{proof}

Section~\ref{sec:logical_tree_construction}
describes the tree expansion, including how budgets are partitioned during the core-tree expansion and how the supplementary leaves are assigned at the boundary of the core tree.
We now make these two partition rules explicit, as they will be used in the proofs below.

First, consider a physical node
$v$
that hosts
$t$
logical nodes belonging to the same candidate at core depth
$d$,
where
$0\le d<H$,
and has total budget
$x$
for these
depth-$d$
logical nodes.
As described in Section~\ref{sec:logical_tree_construction},
the node creates
$\ell\cdot t$
logical children of depth
$d+1$.
Each hosted logical node is the logical parent of exactly
$\ell$
of these children.
The budget
$x$
is split among these
$\ell\cdot t$
logical children as follows.
Order the
$\ell\cdot t$
children arbitrarily as children
$1,\ldots,\ell\cdot t$.
For the
$j$th child, where
$1\le j\le \ell\cdot t$,
set
$x_j$,
the budget assigned to the
$j$th child, to be
\[
x_j
=
\min\{\ell^{H-d},\max\{0,x-(j-1)\ell^{H-d}\}\}.
\]

Second, consider a physical node
$v$
that hosts
$t$
logical nodes belonging to the same candidate at core depth
$H$,
and has total budget
$x$
for these
depth-$H$
logical nodes.
As described in Section~\ref{sec:logical_tree_construction},
the budget
$x$
is partitioned among these
$t$
hosted logical nodes, and the resulting values determine how many supplementary leaves each hosted logical node creates.
Order the
$t$
hosted
depth-$H$
logical nodes arbitrarily as nodes
$1,\ldots,t$.
For the
$j$th hosted logical node, where
$1\le j\le t$,
set
$y_j$,
the number of supplementary leaves assigned to the
$j$th hosted logical node, to be
\[
y_j
=
\min\{\ell,\max\{0,x-(j-1)\ell\}\}.
\]
The
$j$th hosted logical node is the logical parent of exactly
$y_j$
supplementary leaves.

Next, in
Lemma~\ref{lem:exact-volume-physical-budget-flat},
we verify that the construction has the claimed properties in the degenerate branching case
$\ell=1$.
Intuitively, in this case, the core tree is simply a path, and the construction adds one final leaf to obtain exactly
$N$
logical nodes.

\begin{lemma}\label{lem:exact-volume-physical-budget-flat}
For every
$n \ge 2$
and
$\ell = 1$,
the construction as described in Section~\ref{sec:logical_tree_construction}
(as implemented in Algorithm~\ref{alg:tree_election}) produces an almost complete
$\ell$-ary logical tree with exactly
$N$
logical nodes.
That is,
for every depth
$0\le d<H$,
each logical node of depth
$d$
has exactly
$\ell$
logical children, and every logical node of depth
$H$
has at most
$\ell$
logical children.
The number of supplementary leaves added to the core tree, namely at depth
$H+1$,
is exactly the initial budget
$\LeafNumber(\ell,N)$.
\end{lemma}

\begin{proof}
By Lemma~\ref{lem:exact-core-height},
$H \ge 0$.
So our core tree height is well defined (it may include only the root, when
$H = 0$).
By the definition of
$H$
for
$\ell=1$,
we have
$H=N-2$.
By the definition of
$T_H$
for
$\ell=1$,
we have
$T_H=H+1=N-1$.
By the definition of
$\LeafNumber(1,N)$,
we get
$\LeafNumber(1,N) = N - T_H = 1$.

Due to the branching parameter
$\ell = 1$,
the core tree is a path of height
$H$,
and therefore contains exactly
$T_H = H + 1 = N-1$
logical nodes.
By induction, the budget is exactly
$1$.
Indeed, the root starts with budget
$\LeafNumber(1,N) = 1$.
At every core depth
$0 \le d < H$,
there is one hosted logical token and one logical child.
The rule defining
$x_j$
has only
$j=1$
and gives
\[
x_1
=
\min\left\{
\ell^{H-d},
\max\left\{
0,
x-(j-1)\ell^{H-d}
\right\}
\right\}
=
\min\left\{
1,
\max\left\{
0,
1
\right\}
\right\}
=
1,
\]
since
$x = 1$
by induction.

Thus, the budget
$1$
is passed along the unique path until the unique core leaf at depth
$H$.
At depth
$H$,
there is one hosted logical token and budget
$1$.
The rule defining
$y_j$
has only
$j=1$
and gives
\[
y_1
=
\min\left\{
\ell,
\max\left\{
0,
x-(j-1)\ell
\right\}
\right\}
=
\min\left\{
1,
\max\left\{
0,
1
\right\}
\right\}
=
1.
\]
Hence, the unique core leaf creates exactly one additional leaf.
Therefore, the total number of logical nodes is
\[
T_H+\LeafNumber(1,N)
=
(N-1)+1
=
N.
\]

Moreover, for every depth
$0\le d<H$,
the unique logical node of depth
$d$
creates exactly one logical child.
The unique logical node of depth
$H$
creates exactly one supplementary leaf.
Thus, the logical tree is an almost complete
$1$-ary tree.
\end{proof}

We now verify the claimed properties of the construction in the genuine branching case
$2\le \ell\le N$.
Intuitively, here, the core tree is a perfect
$\ell$-ary tree, and the remaining
$\LeafNumber(\ell,N)$
leaves are distributed using the bundled physical budget.

\begin{lemma}\label{lem:exact-volume-physical-budget-tree}
For every
$n \ge 2$
and every integer
$2 \le \ell \le N$,
the construction as described in Section~\ref{sec:logical_tree_construction}
(as implemented in Algorithm~\ref{alg:tree_election}) produces an almost complete
$\ell$-ary logical tree with exactly
$N$
logical nodes.
That is,
for every depth
$0\le d<H$,
each logical node of depth
$d$
has exactly
$\ell$
logical children, and every logical node of depth
$H$
has at most
$\ell$
logical children.
The number of supplementary leaves added to the core tree, namely at depth
$H+1$,
is exactly the initial budget
$\LeafNumber(\ell,N)$.
\end{lemma}

\begin{proof}
By Lemma~\ref{lem:exact-core-height},
$H$
is well defined,
$H \ge 0$.
So our core tree height is well defined (it may include only the root, when
$H = 0$).
Additionally, by the same Lemma, we have
\begin{equation}\label{eq:exact-core-maximality}
T_H<N\le T_{H+1}.
\end{equation}

Therefore,
$1\le \LeafNumber(\ell,N)=N-T_H$.
Also, since
$T_{H+1}=T_H+\ell^{H+1}$,
Equation~\eqref{eq:exact-core-maximality} implies
\[
\LeafNumber(\ell,N)
=
N-T_H
\le
T_{H+1}-T_H
=
\ell^{H+1}.
\]
Thus,
\begin{equation}\label{eq:exact-leaf-capacity}
1\le \LeafNumber(\ell,N)\le \ell^{H+1}.
\end{equation}

We prove the following invariant.
For every core depth
$0 \le d \le H$,
for every physical node
$v$
that hosts
$t$
logical tokens that correspond to the same candidate and to depth
$d$,
and has total budget
$x$
for those
depth-$d$
tokens, we have
\begin{equation}\label{eq:physical-budget-invariant}
0\le x\le t\ell^{H-d+1}.
\end{equation}

We prove the invariant by induction on the core depth
$d$.
At the root, we have
$t=1$,
$d=0$,
and
$x=\LeafNumber(\ell,N)$.
By Equation~\eqref{eq:exact-leaf-capacity},
\[
x\le \ell^{H+1}=t\ell^{H-d+1}.
\]
Thus, the invariant holds at core depth
$0$.

Now fix a depth
$d$,
such that
$0 \le d <H$,
and assume that the invariant holds at this depth.
Consider a physical node
$v$
that hosts
$t$
logical tokens of depth
$d$
and has total budget
$x$
for these tokens.
By the induction hypothesis,
$0\le x\le t\ell^{H-d+1}$.
The construction creates
$\ell t$
logical children and assigns budgets
$x_1,\dots,x_{\ell t}$
according to
\[
x_j
=
\min\left\{
\ell^{H-d},
\max\left\{
0,
x-(j-1)\ell^{H-d}
\right\}
\right\}.
\]
For every
$j$,
this definition gives
$0\le x_j\le \ell^{H-d}$.
Moreover, since
$x\le t\ell^{H-d+1}
=
\ell t\cdot \ell^{H-d}$,
the greedy assignment exhausts the entire budget.
Therefore,
$\sum_{j=1}^{\ell t}x_j=x$.
Each child is at depth
$d+1$.
The budget bound required for one logical token at depth
$d+1$
is
\[
\ell^{H-(d+1)+1}
=
\ell^{H-d}.
\]
Thus, every child produced by the construction satisfies the required bound before bundling.

If several children arrive at the same physical node at depth
$d+1$,
the construction bundles together the children that correspond to the same candidate and to the same depth by summing their token counts and their budgets.
Suppose the incoming bundles that correspond to this candidate and to this depth are indexed by
$r$,
where bundle
$r$
contains
$t_r$
logical tokens and budget
$x_r$.

By the previous paragraph, each incoming bundle satisfies
$x_r\le t_r\ell^{H-d}$.
After bundling, the new token count is
$t'=\sum_r t_r$,
and the new budget is
$x'=\sum_r x_r$.
Hence,
\[
x'
=
\sum_r x_r
\le
\sum_r t_r\ell^{H-d}
=
t'\ell^{H-d}.
\]
Since
\[
H-(d+1)+1=H-d,
\]
this is exactly the invariant at depth
$d+1$.
Therefore, the invariant is preserved from depth
$d$
to depth
$d+1$.
By induction, the invariant holds at every core depth
$d$,
for
$0 \le d \le H$.

Let us now consider the supplementary leaves added to the core tree.
Suppose a physical node
$v$
hosts
$t$
logical tokens of depth
$H$
and has total budget
$x$
for these tokens.
By the invariant,
i.e.,
Equation~\ref{eq:physical-budget-invariant}, we have
\[
0\le x\le t\ell.
\]
The construction orders the
$t$
hosted logical tokens arbitrarily and assigns values
$y_1,\dots,y_t$
according to
\[
y_j
=
\min\left\{
\ell,
\max\left\{
0,
x-(j-1)\ell
\right\}
\right\}.
\]
For every
$j$,
this definition gives
$0\le y_j\le \ell$.
Since
$x\le t\ell$,
the greedy assignment exhausts the entire budget.
Therefore,
$\sum_{j=1}^t y_j=x$.
The
$j$th
hosted logical token creates exactly
$y_j$
additional leaf tokens.
Thus, every logical token at depth
$H$
creates at most
$\ell$
logical children.
For depths
$d<H$,
the construction makes every logical token create exactly
$\ell$
logical children.
Therefore, the logical tree is an almost complete
$\ell$-ary
tree.

It remains to count the number of logical nodes.
The core tree has exactly
$T_H$
logical nodes.
The budget-splitting rule preserves the sum of all budgets from one core depth to the next.
The bundling step also preserves this sum, because budgets are only summed over tokens of the same depth.
At depth
$H$,
the construction creates exactly
$x$
supplementary leaves from every physical bundle with budget
$x$.
Therefore, the total number of supplementary leaves created at depth
$H+1$
is exactly the initial root budget,
namely
$\LeafNumber(\ell,N)$.
Hence, the total number of logical nodes is
\[
T_H+\LeafNumber(\ell,N)
=
T_H+(N-T_H)
=
N.
\]
\end{proof}

Combining the two cases
$\ell = 1$
of Lemma~\ref{lem:exact-volume-physical-budget-flat},
and
$2 \le \ell \le N$
of Lemma~\ref{lem:exact-volume-physical-budget-tree}, 
we obtain the proof of Lemma~\ref{lem:exact-volume-physical-budget}.
\qed